\documentclass[11pt]{scrartcl}
\usepackage[a4paper, total={16cm, 24cm}]{geometry}

\usepackage{multicol}
\usepackage{tocloft}

\usepackage[small]{caption}
\usepackage{subcaption}

\usepackage{algorithm}
\usepackage[table,dvipsnames]{xcolor}
\usepackage{todonotes}
\usepackage{pdflscape}
\usepackage{adjustbox}

\usepackage{natbib}
\usepackage[pagebackref]{hyperref}
\hypersetup{
	pdfencoding=auto,
	psdextra,
	colorlinks=true,
	citecolor=green!50!black,
	linkcolor=red!60!black,
	urlcolor=purple!70!black
}

\renewcommand{\epsilon}{\varepsilon}

\usepackage{authblk}

\author[1]{Benjamin Cookson}
\author[2]{Eva Deltl}
\author[3]{Yeeseok Oh}
\affil[1]{University of Toronto, \texttt{bcookson@cs.toronto.edu}}
\affil[2]{TU Clausthal, \texttt{eva.deltl@tu-clausthal.de}}
\affil[3]{The University of Tokyo, \texttt{oh@g.ecc.u-tokyo.ac.jp}}

\usepackage{url}
\usepackage{microtype}
\usepackage[utf8]{inputenc}
\usepackage{graphicx}
\usepackage{amsmath}
\usepackage{amsthm}
\usepackage{amssymb}
\usepackage{mathtools}
\usepackage{nicefrac}
\usepackage{booktabs}
\usepackage{enumitem}
\usepackage{tabulary}
\usepackage{tabularx}

\usepackage{wrapfig}

\usepackage[nameinlink]{cleveref}
\usepackage{stmaryrd}
\usepackage{thm-restate}
\usepackage{tikz}
\usetikzlibrary{arrows.meta}
\usepackage{tcolorbox}
\newcommand{\lpcmt}[1]{{\color{NavyBlue}\texttt{// #1}}}

\newtheoremstyle{sfthm}
	{\topsep}
	{\topsep}
	{\itshape}
	{}
	{\sffamily\bfseries}
	{}
	{.5em}
	{}
\theoremstyle{sfthm}
\newtheorem{theorem}{Theorem}[section]
\newtheorem{lemma}[theorem]{Lemma}

\newtheoremstyle{sfdef}
{\topsep}
{\topsep}
{}
{}
{\sffamily\bfseries}
{}
{.5em}
{}
\theoremstyle{sfdef}

\title{Improved Lower Bounds for Proportionally Fair Clustering}

\date{}

\begin{document}

\maketitle

\begin{abstract}
\begin{center}
	\textbf{\textsf{Abstract}} \smallskip
\end{center}
We study proportionally fair clustering, where a set of $k$ centers must be chosen from a metric space to represent $n$ agents, and no sufficiently large group of agents should be collectively underrepresented. One of the central notions of fairness in this setting is the \emph{$\alpha$-core}. The existence of clusterings in the $(1+\sqrt{2})$-core was established by \citet{ChenEtal2019}, who also showed instances where the $\alpha$-core is empty for every $\alpha < 2$. Closing this gap has remained an open problem for seven years.
We make progress from the lower-bound side by providing an instance whose $\alpha$-core is empty for every $\alpha < 2.1508$. Our techniques rely on establishing connections between variants of the core, namely the Hare core and the Droop core; reducing the search for optimal empty-core instances to a highly structured family of clustering instances; and using a Mixed Integer Linear Program (MILP) to search for optimal lower-bound instances within this reduced space.
Using this framework, we also determine tight bounds for Droop quota clustering instances with a small number of possible candidate centers and a single center to be selected. For each number of centers $m \in \{3,4,5,6\}$, we give the exact threshold $\alpha_m^*$ such that an $\alpha_m^*$-core clustering always exists, while for every $\alpha < \alpha_m^*$ there is an instance with $m$ centers whose $\alpha$-core is empty. Although these values were originally found through computer-aided search, we also provide direct proofs that do not rely on MILP certificates.
\end{abstract}

\section{Introduction}

Clustering is a fundamental algorithmic task in unsupervised machine learning. Formally, clustering tries to choose a set of $k$ ``centers'' in a metric space to represent a set of $n$ data points, assigning each point to a nearby center. Classical objectives such as $k$-means and $k$-medians attempt to solve this problem by optimizing a global objective function. While these methods are widely used and have been highly successful in many machine learning applications, there are settings where this approach is not ideal.

In many metric clustering problems, each data point corresponds to a person who wants to be \emph{represented} by a center as close to them as possible. One classic example is the facility location problem, where a city must decide where to place schools or hospitals in a way that best serves its constituents. In this problem, each data point corresponds to the location of one resident, and each resident wants to be as close as possible to one of the new facilities.

Another problem of this sort arises when attendees of a talk submit questions to be put to the speaker, and a moderator must select $k$ of them to ask. For each attendee, they would like for their submitted question to be well-represented by one of the selected questions. If their own question is not selected, they would at least like a question that is thematically or semantically similar to be chosen. While this may not immediately seem like a problem that can be represented in a metric space, recent research in the areas of natural language processing and large language models has shown that natural language phrases can be accurately embedded into a metric space so that distance in this space corresponds to semantic or preferential similarity (see \citep{blair2026embeddings} for example).

In problems of this kind, it is known that optimization-based algorithms such as $k$-means and $k$-medians may fail to choose a set of centers that \emph{fairly represents} the people involved.
To remedy this, \citet{ChenEtal2019} introduced a notion of fairness for clustering grounded in the notion of proportional representation from social choice theory. A clustering $X$ with $k$ centers meets this definition of fairness if no sufficiently large group of points is provably \emph{underrepresented} by the clustering. Formally, a group of $\lceil n/k \rceil$ data points is underrepresented by a factor of $\alpha$ if there exists an alternative center that is more than $\alpha$ times closer to each point in the group than any of the $k$ centers chosen by the clustering. If a clustering admits no such underrepresented group for a given value of $\alpha$, then it is said to be in the $\alpha$-core (note that \citet{ChenEtal2019} originally referred to this property as ``$\alpha$-proportional fairness''. The term ``core'' for this property was first used in \citet{ebadian2024boosting}, and is the terminology we use in this work). In their seminal work, \citet{ChenEtal2019} were able to establish that, when the data points lie in any arbitrary metric space, a clustering in the $(1+\sqrt{2})$-core always exists. To complement this, they also provided an instance in which no clustering in the $\alpha$-core exists for any $\alpha < 2$. Together, these results left a gap of $[2,\,1+\sqrt{2}] \approx [2,\,2.414]$.

In the subsequent $7$ years since the paper of \citet{ChenEtal2019}, there have been many papers~\citep{MichaShah2020, AzizEtal2024, KellerhalsPeters2024, cookson2025unifying} which built on the proportionally fair clustering model, iterating on the definitions, and proving better upper and lower proportional fairness bounds for special cases of metric spaces. However, for general metric spaces, there have been no improvements to the $[2,\, 1+\sqrt{2}]$ gap.

\subsection{Our Contributions}
We make progress on closing this gap from the lower-bound side. The main theorem of this paper provides a new lower bound for proportionally fair clustering in general metric spaces:

\begin{theorem}\label{thm:T=37}
	There exists a clustering instance where the $\alpha$-core is empty for every $\alpha<2.1508$.
\end{theorem}

The primary mechanism used in the proof of \Cref{thm:T=37} is computer-aided search. Specifically, we construct a program that takes in a partial clustering instance and solves a series of linear programs to construct a full clustering instance where the $2.1508$-core is empty. The ``proof'' of \Cref{thm:T=37} takes the form of a certificate generated by a sequence of linear programs (see appendix for a link to a GitHub repository containing the code for this program). The majority of the paper is devoted to explaining the structure of this program, and proving why its output guarantees this lower bound.

More specifically, our main insight relies on a connection established recently by \citet{kellerhals2026proportional} between proportionally fair clustering under the Hare quota (the traditional version studied in previous clustering papers) and the Droop quota (a more restrictive definition of proportional fairness most prominently studied by \citet{aronov2020beta} under the name of the $\beta$-plurality problem). In \Cref{sec:droop-hare}, we build on the connections developed by \citet{kellerhals2026proportional} by remarking that a lower bound for the core under the Droop quota implies the same lower bound under the Hare quota. Next, in \Cref{sec:reduce}, we show that by focusing our search on a specific family of highly structured Droop-quota instances, we can efficiently search through all relevant instances using a computer-aided approach. This search identifies instances whose core is empty for $\alpha > 2$.

In addition to using these techniques to prove the lower bound in \Cref{thm:T=37}, we also leverage them to prove tight bounds on the core with the Droop quota and $k=1$ when the number of possible centers is small.
For instances with $m \le 6$ centers, we determine the exact threshold $\alpha_m^*$ such that an $\alpha_m^*$-core clustering always exists, while for every $\alpha<\alpha_m^*$ there is an instance with $m$ centers whose $\alpha$-core is empty.
\Cref{tab:best-alphas-contri} shows the exact optimal value of $\alpha$ attainable for instances with different numbers of centers. Additionally, for each of these tight bounds on small instances, we do not only provide a proof based on a computer-generated certificate. We also prove them directly, giving interpretable lower-bound constructions and proofs for these bounds. Most notably, for $m=6$, we are able to show a relatively simple instance whose $\alpha$-core is empty for any $\alpha < 2.1019$, improving on the lower bound established by \citet{ChenEtal2019}. This complements our stronger LP-based lower bound of $2.1508$, and may be more useful to readers than the black-box LP certificate.

\begin{table}[ht]
	\centering
	\begin{tabular}{cc}
		\toprule
		$m$ & $\alpha_m^*$ \\
		\midrule
		3 & 2 \\
		4 & 2 \\
		5 & 2.0775 \\
		6 & 2.1019 \\
		\bottomrule
	\end{tabular}
	\caption{Tight $\alpha$-values for instances with $m$ centers for $m \in \{3, 4, 5, 6\}$.}
	\label{tab:best-alphas-contri}
\end{table}

\subsection{Related Work}\label{sec:related-work}

\paragraph{Proportionally fair centroid clustering.} As stated previously, \citet{ChenEtal2019} introduced the notion of proportional fairness to the clustering setting, and established the upper and lower bounds of $1 + \sqrt{2}$ and $2$ for general metric spaces, with the upper bound being achieved using the \emph{Greedy Capture} algorithm. \citet{MichaShah2020} showed that in specific metric spaces, Greedy Capture achieved a better approximation ratio than $1 + \sqrt{2}$. Specifically, in Euclidean spaces with the $L^2$ norm, Greedy Capture achieves a $2$-approximation of the core. Additionally, they showed that in metric spaces induced by trees, a perfectly fair clustering in the $1$-core always exists.

Several papers have suggested notions other than the core that aim to capture proportional fairness in the clustering setting. \citet{AzizEtal2024} introduce an alternate fairness definition for clustering which they term ``Proportionally Representative Fairness'' (PRF), which is directly inspired by notions of Proportional Justified Representation (PJR) from multi-winner voting. They introduce an algorithm that achieves PRF while also achieving a constant approximation of the core. \citet{KellerhalsPeters2024} establish several further connections to multi-winner voting by adapting additional notions to the clustering setting.

Algorithms for proportionally fair clustering have been used and adapted for many interesting applications. \citet{ebadian2024boosting} show how clustering can be used to model the process of sortition, and introduce an adaptation of the greedy capture algorithm that can randomly select a panel of citizens in a way that is proportionally fair. \citet{aziz2026fair} showed how algorithms from proportionally fair clustering can be applied to find fair placements of transit stops.

\paragraph{$\beta$-plurality problem.} Independently of the development of proportionally fair clustering, several papers have also studied the $\beta$-plurality problem. We describe this problem formally in \Cref{sec:prelims}, but it is essentially the proportionally fair clustering problem, where only a single center is chosen, and there is a smaller threshold (the Droop quota) for when agents can deviate. \citet{aronov2020beta} introduced this model and showed bounds on the problem for special cases such as when the possible center can be selected from any point in $\mathbb{R}^d$. \citet{filtser2020plurality} showed an upper bound of $1 + \sqrt{2}$ and a lower bound of $2$ under general metric spaces. \citet{kellerhals2026proportional} formally established the connection between the $\beta$-plurality problem and centroid clustering. Most notably for our work, they observed that these two problems have the same bounds in general metric spaces, and suggested that a reasonable strategy for improving these bounds would be to focus on the ``simpler'' $\beta$-plurality problem. This is exactly the approach we take in this paper.

\paragraph{Clustering with other preferences.} The version of fair clustering discussed in this paper, where the agent's preferences are determined by their distances to the selected centers, is often referred to as \emph{centroid} clustering. Proportional fairness has also been applied to \emph{non-centroid} clustering, where rather than selecting a set of $k$ centers in the metric space to map the agents to, agents are directly partitioned into $k$ different groups, and have preferences dependent on the other agents who appear in their group. This setting was first introduced by \citet{caragiannis2024proportional} (though special cases of this problem were also studied previously in papers such as \citep{arkin2009geometric}, calling the problem the ``Geometric stable roommates problem''). Most notably, \citet{caragiannis2024proportional} showed that when each agent's preferences are derived from their maximum distance to any other agent in their assigned cluster, then a slight modification of the Greedy Capture algorithm achieves a $2$-approximation of the core. Later, \citet{bredereck2025core} complemented this with a lower bound of $2^{1/5}$ in this setting. These are the best upper and lower bounds known for non-centroid clustering under general metric spaces. \citet{cookson2025unifying} introduced a single model that unified centroid and non-centroid clustering. They highlighted the subtle differences between the Greedy Capture algorithms that achieved core approximations in each setting, and introduced a new algorithm that could find a clustering in the $3$-core for a new family of preferences where agents take both their assigned center and the other agents in their cluster into consideration.

\paragraph{Computer searches for improved bounds in social choice.} Finally, we highlight some other papers which use computer-aided search and (Mixed Integer) Linear Programs (MILP) to tackle various problems of similar flavor in social choice theory. \citet{peters2025core} recently showed that a committee in the core always exists for approval-based committee elections with small numbers of candidates. The key argument to do this was based on a reduction to a linear program that could be solved using a computer for small instances. \citet{berker2026edge} also study the core in approval-based committee elections, and search for core-violating instances using a MILP formulation, resolving several open problems regarding the compatibility of the core with other notions from this field.

\section{Preliminaries}\label{sec:prelims}
For a positive integer $t$, denote $[t] = \{1,2, \ldots, t\}$ and $[t]_0 = \{0,1, \ldots, t-1\}$.

\paragraph{Clustering instances.}
A \emph{clustering instance} $I=(N,C,d,k)$ consists of a set of $n$ \emph{agents} $N=[n]$, a set of $m$ \emph{centers} $C=\{c_0,c_1,\ldots,c_{m-1}\}$, a \emph{(pseudo-)metric} $d\colon(N\cup C)\times(N\cup C)\to\mathbb{R}_{\ge 0}$ satisfying the triangle inequality, and a positive integer $k\le m$ specifying how many centers are to be \emph{chosen}.  A \emph{clustering} is a subset $X\subseteq C$ with $|X|=k$.  For agent $i\in N$, the \emph{loss} under $X$ is $\ell_i(X)=\min_{c\in X}d(i,c)$.

When several agents occupy the same point in the metric space, it is natural to represent them compactly as a single agent with fractional weight. A \emph{weighted instance} augments a clustering instance with a weight function $w\colon N\to\mathbb{Q}_{\ge 0}$ normalized so that $\sum_{i\in N}w_i=1$, where $w_i$ captures the fraction of the population located at point $i$. Any weighted instance can be thought of as representing a standard unweighted one by scaling to a common denominator; the two formulations are equivalent. Weighted instances are useful not only because they provide a more compact representation of instances, but also as the continuous nature of the weight function allows us to search over instances via computer-aided search more easily. We adopt weighted instances as the default throughout this paper, writing $I=(N,C,d,k,w)$ in full and $I=(N,C,d,w)$ when $k=1$. If an instance is specified without explicit weights, we assume $w_i=1/n$ for all $i\in N$.

\paragraph{The core and proportional fairness.}
Under the \emph{Hare quota}, any group of agents with total weight $\ge 1/k$ is collectively entitled to one center.  A clustering $X$ is in the \emph{$\alpha$-core} (for $\alpha\ge 1$) if no entitled group can jointly $\alpha$-improve: there is no center $y \in C$ and set $S\subseteq N$ with $\sum_{i\in S}w_i\ge 1/k$ such that
\[
  \alpha\cdot d(i,y) < \ell_i(X)\quad\text{for all }i\in S.
\]
We say such a pair $(S,y)$ is a \emph{blocking coalition} and that the agents in $S$ \emph{$\alpha$-deviate} (from $X$) \emph{to} $y$.

\paragraph{The Droop quota and Droop core.}
The \emph{Droop quota} entitles any group of agents with total weight $>1/(k+1)$ to one center.  The \emph{$\alpha$-Droop core} is defined identically to the $\alpha$-core but with the Droop quota replacing the Hare quota. Thus, $X$ is in the $\alpha$-Droop core if there is no center $y\in C$ and no set $S\subseteq N$ with $\sum_{i\in S}w_i>1/(k+1)$ such that $(S,y)$ is a blocking coalition.  For $k=1$ the Droop quota reduces to a \emph{strict majority} of weight ($\sum_{i\in S}w_i>1/2$), so a clustering is in the $\alpha$-Droop core exactly when no strict-majority group can all $\alpha$-deviate to a single alternative center.

As originally pointed out by \citet{kellerhals2026proportional}, the $\alpha$-Droop core with $k=1$ is closely related to the $\beta$-plurality problem \citep{aronov2020beta}. The key distinction is that the $\beta$-plurality problem typically allows \emph{any} point in the underlying metric space to be chosen as a center, in particular this implies that $N \subseteq C$. In our work, we will not assume this stronger restriction on instances, and allow $C$ to be an arbitrary set of points.

\section{A formal connection between Droop quota and Hare quota}\label{sec:droop-hare}
To establish a new lower bound for proportionally fair clustering, it suffices to find an instance whose $\alpha$-core is empty for some $\alpha > 2$ and some value of $k$. Thus, one would assume the easiest place to start to be the simplest possible setting, namely $k=1$. When $k=1$, the problem of clustering reduces to selecting a single center to represent all agents. However, since we are primarily interested in lower bounds for Hare quota clustering, the most common formulation of the problem, this setting is trivial: when $k=1$, any deviating coalition must have total weight at least $1/k = 1$. Since all weights sum to $1$, the coalition must include all agents. Hence, every center that is (weakly) Pareto optimal will be in the core, and a center in the $1$-core must always exist.

However, this trivial existence no longer holds when we instead use the Droop quota. As first observed by \citet{kellerhals2026proportional}, proportionally fair clustering using the Droop quota has the same upper and lower bounds of $[2, 1+\sqrt{2}]$ as Hare quota clustering. In the Droop quota case, when $k=1$, any group of agents with total weight strictly greater than $1/(1+1) = 1/2$ can deviate.

Here, we formally prove a strong connection between these two quotas. We show that any lower bound for Droop quota instances directly implies the same lower bound for Hare quota instances.

\begin{theorem}\label{thm:core-equivalence}
	Let $\alpha$ be a constant. Then, $\alpha$ is a lower bound for the Droop core if and only if $\alpha$ is a lower bound for the (Hare) core.
\end{theorem}
\begin{proof}
	If $\alpha$ is a lower bound for the Hare core, then it trivially is a lower bound for the Droop core, since the Droop quota is at most the Hare quota. Thus, any blocking coalition of Hare size is also a blocking coalition of Droop size.

	Conversely, let $I=(N,C,d,k,w)$ be a weighted instance for an $\alpha$-lower bound for the Droop core.
	Since any weighted instance is equivalent to a uniform-weight one (see \Cref{sec:prelims}), assume without loss of generality that $w_i = 1/n$ for all $i \in N$.
	The Droop quota is then $\ell = \lfloor \frac{n}{k+1} \rfloor + 1 > \frac{n}{k+1}$ agents, equivalently total weight $\ge \ell/n > 1/(k+1)$.
	Construct a weighted instance $I' = (N',C',d',k',w')$ by taking $\ell$ copies of $(N,C)$, each infinitely far away from each other, giving each agent weight $\frac{1}{n\ell}$. Let $k'=n$. The Hare quota of $I'$ is $1/k' = 1/n$.
	Consider any choice of $k'=n$ centers in $I'$. Since there are $\ell > \frac{n}{k+1}$ copies of $I$ and $n$ selected centers, at least one copy has at most $k$ centers being selected from it. Otherwise, each copy would receive at least $k+1$ centers, requiring more than $\ell \cdot (k+1) > n$ centers.
	In that copy, the original Droop lower bound instance applies, as each agent in the copy attains her loss by a center in the copy. Since the $\alpha$-Droop core is empty in the original instance, there is a set of agents with total weight at least $\frac{\ell}{n\ell} = \frac{1}{n} = \frac{1}{k'}$ who $\alpha$-deviate. This meets the Hare quota of $I'$, so this set forms a Hare blocking coalition.
\end{proof}

Most notably, \Cref{thm:core-equivalence} tells us that if we can find an instance with $k=1$ where the $\alpha$-Droop core is empty for some $\alpha > 2$, then there exists an instance (with $k > 1$) where the $\alpha$-(Hare) core is also empty. Thus, for the remainder of the paper, we will exclusively focus on Droop quota instances with $k=1$.

For the interested reader, \Cref{thm:core-equivalence} has other implications that are less relevant to the main results of this paper. One may notice that the proof of the theorem does not use the fact that agents' preferences are derived from a metric space. In fact, this theorem holds generally for a much larger class of voting instances. In the appendix~(\Cref{sec:ABC}), we provide a detailed discussion of the implications of this theorem (and some extensions of it) to Approval-Based Committee voting.

This theorem also allows us to easily derive new results for well-studied special cases of clustering, such as the case where $N = C$ (the set of agents and the set of centers are exactly the same). In the appendix~(\Cref{sec:N=C}), we leverage an instance of \citet{filtser2020plurality} along with \Cref{thm:core-equivalence} to show a lower bound of $2$ for clustering in this setting in general metric spaces. A similar result has also been proven by \citet{Trinh2024}, but we believe our proof is quite simple in comparison.

\section{Deviation Graphs and Deviation Cycles}\label{sec:deviation-graphs}

We will begin by discussing the main technical tool we use to analyze clustering instances and derive our lower bounds. The \emph{deviation graph} $G_\alpha(I)$ of an instance $I$ provides the central tool for reasoning about the $\alpha$-Droop core with $k=1$.

\paragraph{The deviation graph.}
Fix $\alpha\ge 1$, $k=1$ and a weighted instance $I=(N,C,d,w)$. The \emph{deviation graph} $G_\alpha(I)$ is the directed graph with vertex set $C$ and edge set $E$ defined by
\[
  (x,y)\in E\;\iff\;\sum_{i\in N:\,\alpha\cdot d(i,y)<d(i,x)} w_i\;>\;\frac{1}{2}.
\]
An edge $(x,y)$ encodes that a strict-majority weight of agents would $\alpha$-deviate to $y$ from $x$. When the instance $I$ is clear from context, we write $G_\alpha$ for $G_\alpha(I)$.

Rather than reasoning directly about which subsets of agents can form blocking coalitions, we encode all deviation information into a directed graph on the set of centers, and reduce the question of core existence to a simple structural property of that graph: the presence or absence of a vertex with no outgoing edges.

\paragraph{Cycles and deviation sets.}
A key observation that will drive most of our results is regarding the existence of \emph{cycles} in these deviation graphs. Particularly, we say that a \emph{sink} of $G_\alpha(I)$ is a center $c^*$ with no outgoing edges; by definition, any center that is a sink is in the $\alpha$-Droop core, as no majority-weighted coalition of agents would want to $\alpha$-deviate from that center to some other center in the instance. Therefore, if some instance $I$ has an empty $\alpha$-Droop core for some value of $\alpha$, then its deviation graph $G_\alpha(I)$ will contain no sinks. Any directed graph with no sinks must contain at least one directed cycle.

We write a directed cycle of length $T$ as $c_0c_1\cdots c_{T-1}c_0$, using the index set $[T]_0=\{0,1,\ldots,T-1\}$ with all arithmetic on indices taken modulo $T$.  For each $t\in[T]_0$, define the \emph{($\alpha$-)deviation set}
\[
  S_t=\{i\in N:\alpha\cdot d(i,c_t)<d(i,c_{t-1})\},
\]
that is, the agents who would $\alpha$-deviate from $c_{t-1}$ to $c_t$.

If some deviation graph contains a cycle $c_0c_1\cdots c_{T-1}c_0$, then since $(c_{t-1},c_t)\in E$ for every $t \in [T]_0$, we also have that $\sum_{i\in S_t}w_i>1/2$. This implies that consecutive deviation sets overlap: $\sum_{i\in S_t}w_i+\sum_{i\in S_{t+1}}w_i>1$, so $S_t\cap S_{t+1}\neq\emptyset$. Throughout the paper, when reasoning about cycles that appear in deviation graphs, we will use $i_t$ to denote the agent that is guaranteed to exist in $S_t\cap S_{t+1}$ for each $t$.

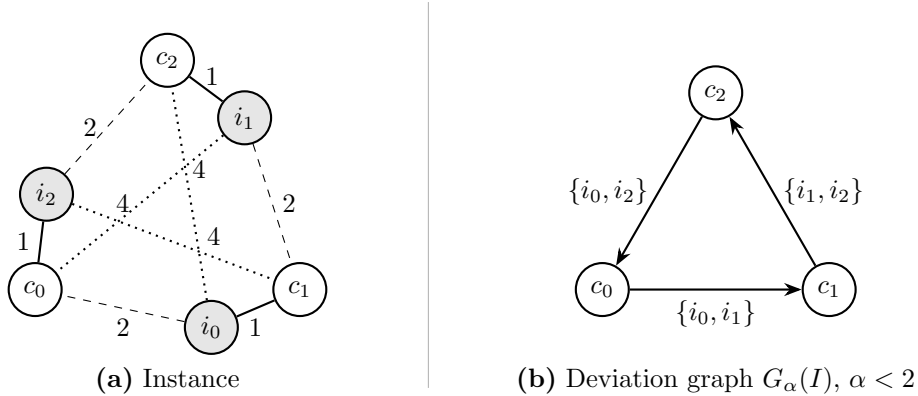
Figure~\ref{fig:deviation-example} illustrates these definitions with a concrete example. Let $I = (\{i_0,i_1,i_2\}, \{c_0,c_1,c_2\}, d, w)$ be the instance~\textbf{(a)}, where $w_{i_0}=w_{i_1}=w_{i_2}=1/3$ and each agent is at distance~$1$ from one center, distance~$2$ from a second, and distance~$4$ from a third. For any $\alpha < 2$, each agent $\alpha$-deviates from their distance-$2$ center to their distance-$1$ center, and from their distance-$4$ center to their distance-$2$ center. Thus, for every center, a strict majority of the agents, of total weight $2/3$, would deviate away from it. The deviation graph $G_\alpha(I)$~\textbf{(b)} is a directed 3-cycle $c_0c_1c_2c_0$ for every $\alpha < 2$. This graph contains no sink, so the $\alpha$-Droop core is empty for every $\alpha < 2$. The deviation sets are $S_1=\{i_0,i_1\}$, $S_2=\{i_1,i_2\}$, $S_0=\{i_0,i_2\}$, shown as edge labels, and the pairwise overlaps are $i_0\in S_0\cap S_1$, $i_1\in S_1\cap S_2$, and $i_2\in S_2\cap S_0$.

\begin{figure}[ht]
\centering
\begin{tikzpicture}[
  cnode/.style={draw, circle, thick, minimum size=7mm, inner sep=1pt, fill=white},
  anode/.style={draw, circle, thick, minimum size=7mm, inner sep=1pt, fill=gray!20},
  >=Stealth,
  font=\small,
]

\node[cnode] (c0) at (0,    0)    {$c_0$};
\node[cnode] (c1) at (3.5,  0)    {$c_1$};
\node[cnode] (c2) at (1.75, 3.03) {$c_2$};

\node[anode] (i0) at (2.33, -0.5)  {$i_0$};
\node[anode] (i1) at (2.77,  2.27) {$i_1$};
\node[anode] (i2) at (0.15,  1.26) {$i_2$};

\draw[thick]         (i0) -- node[below]          {$1$} (c1);
\draw[thick]         (i1) -- node[right=2pt, pos=1] {$1$} (c2);
\draw[thick]         (i2) -- node[left]           {$1$} (c0);

\draw[dashed]        (i0) -- node[below]          {$2$} (c0);
\draw[dashed]        (i1) -- node[right]          {$2$} (c1);
\draw[dashed]        (i2) -- node[left]           {$2$} (c2);

\draw[dotted, thick] (i0) -- node[right, pos=0.3]  {$4$} (c2);
\draw[dotted, thick] (i1) -- node[right, pos=0.25] {$4$} (c0);
\draw[dotted, thick] (i2) -- node[above, pos=0.25] {$4$} (c1);

\node at (1.75, -1.2) {\textbf{(a)} Instance};

\draw[thin, gray!60] (5.2, -1.3) -- (5.2, 3.8);

\begin{scope}[xshift=7.5cm]

\node[cnode] (c0) at (0,   0)   {$c_0$};
\node[cnode] (c1) at (3,   0)   {$c_1$};
\node[cnode] (c2) at (1.5, 2.6) {$c_2$};

\draw[->, thick] (c0) -- node[below] {$\{i_0,i_1\}$} (c1);
\draw[->, thick] (c1) -- node[right] {$\{i_1,i_2\}$} (c2);
\draw[->, thick] (c2) -- node[left]  {$\{i_0,i_2\}$} (c0);

\node at (1.5, -1.2) {\textbf{(b)} Deviation graph $G_\alpha(I)$, $\alpha < 2$};

\end{scope}

\end{tikzpicture}
\caption{An example instance \textbf{(a)} and its deviation graph \textbf{(b)}.
  Solid, dashed, and dotted edges in \textbf{(a)} indicate distances~$1$, $2$, and~$4$, respectively.}
\label{fig:deviation-example}
\end{figure}

To ease into the analysis of deviation graphs and their cycles, we will show a simple alternative proof of the fact that a center in the $(1 + \sqrt{2})$-Droop core always exists when $k=1$ \citep{kellerhals2026proportional}. The following lemma shows a key relationship between the distances of the overlapping agents in a cycle to their associated centers.

\begin{lemma}\label{lem:deviation-step}
	Let $I=(N,C,d,w)$ be a weighted instance and let $c_0c_1\cdots c_{T-1}c_0$ be a directed cycle in $G_\alpha(I)$ for some $\alpha > 1$, with deviation sets $S_t$ and overlap agents $i_t \in S_t \cap S_{t+1}$ for each $t \in [T]_0$. Then
	\[ d(i_{t+1},c_{t+1}) < \frac{\alpha+1}{\alpha(\alpha-1)}\,d(i_t, c_t). \]
\end{lemma}
\begin{proof}
	This follows from the inequalities
	\begin{align*}
		\alpha \cdot d(i_{t+1},c_{t+1}) &<  d(i_{t+1}, c_t) \\
		&\le  d(i_{t+1}, c_{t+1}) + d(i_{t}, c_{t+1}) + d(i_{t}, c_t) \\
		&<  d(i_{t+1}, c_{t+1})  + \left(\frac{1}{\alpha} + 1\right) d(i_t, c_t),
	\end{align*}
	where the first inequality uses $i_{t+1} \in S_{t+1}$, the second uses the triangle inequality, and the third uses $i_t \in S_{t+1}$. Rearranging yields
	$$ d(i_{t+1},c_{t+1}) < \frac{\alpha+1}{\alpha ( \alpha - 1)} d(i_t, c_t).$$
\end{proof}

\Cref{lem:deviation-step} is useful as we can apply it repeatedly around a cycle, starting by bounding $d(i_0, c_0)$ in terms of $d(i_{T-1}, c_{T-1})$, then bounding $d(i_{T-1}, c_{T-1})$ (and thus $d(i_0, c_0)$) in terms of $d(i_{T-2}, c_{T-2})$. Eventually, we can loop back around to bounding all these distances in terms of $d(i_0, c_0)$, which allows us to cancel out the term $d(i_0, c_0)$ and solve for bounds on $\alpha$. This technique is shown formally in the theorem below.

\begin{theorem}\label{thm:general-case}
	There always exists a $(1+\sqrt{2})$-Droop core clustering for $k=1$.
\end{theorem}
\begin{proof}
	Assume for contradiction that the $\alpha$-Droop core is empty for some instance $I=(N,C,d,w)$ and some $\alpha \ge 1+\sqrt{2}$. Then $G_\alpha(I)$ does not contain a sink node, and hence $G_\alpha(I)$ contains a directed cycle. Let $c_0c_1 \ldots c_{T-1}c_0$ be one such cycle in $G_\alpha(I)$. For each $t \in [T]_0$, let $i_t \in S_{t} \cap S_{t+1}$.
	Then applying \Cref{lem:deviation-step} repeatedly, we have
	\begin{align*}
			d(i_{0},c_{0}) &< \left(\frac{\alpha+1}{\alpha ( \alpha - 1)} \right)d(i_{T-1}, c_{T-1}) < \left(\frac{\alpha+1}{\alpha ( \alpha - 1)} \right)^2 d(i_{T-2}, c_{T-2}) \\
			&< \dots < \left(\frac{\alpha+1}{\alpha ( \alpha - 1)} \right)^{T-1} d(i_1, c_1) < \left(\frac{\alpha+1}{\alpha ( \alpha - 1)} \right)^T d(i_0, c_0).
		\end{align*}
	Since $d(i_0,c_0)>0$, this requires $\frac{\alpha+1}{\alpha(\alpha-1)}>1$, which is equivalent to $\alpha < 1+\sqrt{2}$, contradicting $\alpha \ge 1+\sqrt{2}$.
\end{proof}

In the subsequent sections, we will refine our technique for analyzing deviation cycles in order to derive better bounds.

\section{Searching for New Lower Bounds}

\subsection{Reducing the Search Space}\label{sec:reduce}

By  \Cref{thm:core-equivalence}, we can search for lower bounds on Droop quota instances with $k=1$ while still drawing conclusions about the more well-studied Hare quota. This already gives us significantly more simplicity and structure than instances with $k>1$. In this section, we show how this search space can be restricted even further.

These restrictions will ultimately allow us to encode the search for better lower bounds as a Mixed Integer Linear Program (MILP), which found all the improved lower bounds discussed previously.

\paragraph{Reduction 1: Focusing on Hamiltonian Deviation Cycles.} As we highlighted previously, the key to our analysis is looking for cycles in the deviation graphs of instances. Here, we can in fact show that for a given $\alpha$ we can restrict our attention to instances $I$ where $G_\alpha(I)$ contains a Hamiltonian cycle, meaning that some subset of the edges $G_\alpha(I)$ forms a perfect cycle graph through all the centers in $I$.

\begin{lemma}\label{lem:cycles-are-sufficient}
	Suppose there is an instance with $m$ centers such that the $\alpha$-Droop core is empty for $k=1$. Then there is an instance $I'$ with $T \le m$ centers such that $G_{\alpha}(I')$ is a directed cycle of length $T$.
\end{lemma}
\begin{proof}
	Let $I=(N,C,d,w)$ be an instance where the $\alpha$-Droop core is empty. This means that $G_\alpha(I)$ has no sink node, and thus contains at least one directed cycle. Take the shortest cycle $c_0c_1\dots c_{T-1}c_0$ in $G_\alpha(I)$ and let $C' \subseteq C$ be the centers in this cycle.

	W.l.o.g. assume that $N \cap C = \emptyset$, by possibly separating an agent-center $a \in N \cap C$ to $i_a \in N$ and $c_a \in C$ located in the same place.
	Let $I' = (N,C',d,w)$ be a new instance constructed by taking $I$ and removing all centers that were not in the shortest cycle. Since each edge $(x,y)$ of $G_\alpha(I)$ depends only on $d(i,x)$ and $d(i,y)$ for each $i \in N$, removing centers does not change the edges between remaining centers, so $G_\alpha(I')$ contains exactly the edges of $G_\alpha(I)$ induced on $C'$. Furthermore, the induced subgraph on $C'$ has no edges beyond the cycle itself: any additional edge $(c_i, c_j)$ with $j \neq i{+}1 \pmod{|C'|}$ would, combined with the path from $c_j$ back to $c_i$ along the original cycle, yield a cycle of length $1 + (i - j \bmod |C'|) \leq |C'| - 1$, contradicting the minimality of the chosen cycle. Hence $G_\alpha(I')$ is a directed cycle of length $|C'| \leq m$.
\end{proof}

\paragraph{Reduction 2: Restricting the Number of Weighted Agents.} Additionally, we show that if one of the cyclic instances we consider has $m$ centers, then we may assume without loss of generality that it has exactly $m$ weighted agents. The following lemma shows that we can eliminate the number of agents $n$ from consideration, which greatly simplifies the space of instances we need to search over.

\begin{lemma}\label{lem:m-agents-are-sufficient}
	If a clustering instance $I=(N,C,d,w)$ with $m$ centers has a deviation graph $G_\alpha(I)$ containing a directed cycle of length $m$, then there exists an instance $I'$ with $m$ centers and at most $m$ weighted agents such that $G_\alpha(I')$ contains a directed cycle of length $m$.
\end{lemma}
\begin{proof}
	Let $I=(N,C,d,w)$ be an instance with $m$ centers and with $G_\alpha(I)$ containing a directed cycle $c_0c_1\dots c_{m-1}c_0$ of length $m$. Let $S_t$ denote the deviation set from $c_{t-1}$ to $c_t$ for all $t \in [m]_0$. Then we know that all the following are true:
	\begin{equation}
		\begin{array}{ll}
			\sum_{i \in N}w_i = 1 \\
			\sum_{i \in S_t} w_i > \frac{1}{2}, \qquad & \text{for all } t \in [m]_0, \\
			w_i \ge 0,  \qquad & \text{for all } i \in N.
		\end{array}
	\end{equation}
	Hence, set $\epsilon = \min_{t \in [m]_0} \left(\sum_{i\in S_t} w_i - \frac{1}{2} \right) > 0$, and the following linear program is feasible with optimal value at most $1$.

	\begin{equation}\label{eq:LP-weights}
		\begin{array}{lllr}
			\text{minimize}_{w}   & \sum_{i \in N} w_i \\
			\text{subject to}
			&\sum_{i \in S_t} w_i \ge \frac{1}{2} + \epsilon, & & \text{for all } t \in [m]_0, \\
			&w_i \ge 0, & & \text{for all } i \in N.
		\end{array}
	\end{equation}

	Since $\epsilon$ is rational (as weights are rational), all coefficients and right-hand sides of LP~\eqref{eq:LP-weights} are rational. Hence it admits a rational optimal basic feasible solution $(w_i^*)_{i \in N}$. Because the LP has exactly $m$ constraints (excluding non-negativity), this basic feasible solution has at most $m$ nonzero entries, and an optimal total weight $W^* = \sum_{i \in N} w_i^* \le 1$. We delete all agents $i$ with $w_i^*=0$ and normalize the remaining weights by setting $w_i^{**} = w_i^* / W^*$. Let $I' = (N', C, d, w^{**})$ where $N' = \{i \in N : w_i^* > 0\}$. Because $W^* \le 1$, this scaling ensures $\sum_{i \in S_t} w_i^{**} \ge \frac{1}{2} + \epsilon > \frac{1}{2}$ for all $t$, so $G_\alpha(I')$ still contains the directed cycle of length $m$.
\end{proof}

\paragraph{Full Reduction.} Putting together the reductions from \Cref{lem:cycles-are-sufficient,lem:m-agents-are-sufficient}, we can show that, to search for the optimal lower bound for instances with a fixed number $m$ of centers, it suffices to consider instances with at most $m$ centers and exactly $m$ agents, whose deviation graph contains a directed Hamiltonian cycle.

\begin{lemma}\label{lem:full-reduction}
	If there exists a clustering instance $I = (N,C,d,w)$ with $|C| = m$ such that the $\alpha$-Droop core is empty for $k=1$, then there exists an instance $I' = (N',C',d',w')$ with $|C'| \leq m$ and $|N'| \leq |C'|$ where $G_\alpha(I')$ contains a directed cycle of length $|C'|$.
\end{lemma}
\begin{proof}
	We construct $I'$ from $I$ in two steps.

	First, by \Cref{lem:cycles-are-sufficient}, there exists an instance $I'' = (N,C',d,w)$ with $C' \subseteq C$ such that $G_\alpha(I'')$ is a directed cycle of length $|C'|$.

	Then, by \Cref{lem:m-agents-are-sufficient}, the existence of $I''$ implies the existence of an instance $I' = (N',C',d,w')$ with $N' \subseteq N$ such that $G_\alpha(I')$ contains a cycle of length $|C'|$ and $|N'| \leq |C'|$.
\end{proof}

Finally, using the reduction from \Cref{lem:full-reduction}, we can show that for any fixed value of $m$, if we wish to find the highest possible value of $\alpha$ such that there exists an instance with $m$ centers where the $\alpha$-Droop core is empty for $k=1$, then it suffices to focus on instances where there are at most $m$ centers and $m$ agents, and the deviation graph contains a directed Hamiltonian cycle.

Let
\begin{multline*}
	\alpha_m = \sup \{\alpha : \exists \text{ clustering instance } I \text{ with } m \text{ centers and } n \leq m \text{ agents} \\
	\text{s.t. } G_{\alpha}(I) \text{ contains a directed Hamiltonian cycle} \}.
\end{multline*}

\begin{theorem}\label{lem:cycle-search-m}
	For clustering instances with $m$ centers and $k=1$, the optimal bound for the $\alpha$-Droop core to be nonempty is $\alpha_m^* = \max_{T \leq m}\alpha_{T}$. That is, there always exists a center in the $\alpha_m^*$-Droop core and for every $\alpha < \alpha_m^*$, there exists an instance whose $\alpha$-Droop core is empty.
\end{theorem}
\begin{proof} We first observe that, for any $T \leq m$, the supremum defining $\alpha_T$ is not attained; that is, all clustering instances with $T$ centers and $n \le T$ agents do not contain a directed Hamiltonian cycle in $G_{\alpha_T}$. To see this, consider an arbitrary $\alpha$ and an instance $I$ with $T$ centers and $n \le T$ agents where $G_\alpha(I)$ contains a directed Hamiltonian cycle $c_0c_1 \ldots c_{T-1}c_0$
	Then, for each $t \in [T]_0$, we have $\sum_{i \in S_t} w_i > \frac{1}{2}$ and $d(i,c_{t-1}) > \alpha \cdot d(i,c_{t})$ for all $i \in S_t$. Hence, for
	$$\alpha' = \min \left\{ \frac{d(i, c_{t-1})}{d(i, c_t)} : d(i, c_t) > 0, i \in S_t, t \in [T]_0 \right\},$$
	we have $\alpha < \alpha'$. Taking $\alpha'' = \frac{\alpha + \alpha'}{2} \in (\alpha, \alpha')$, we have $d(i,c_{t-1}) > \alpha'' \cdot d(i,c_{t})$ for all $i \in S_t$, so $G_{\alpha''}$ still contains the cycle $c_0c_1 \ldots c_{T-1}c_0$. Since we can always find a strictly larger witness $\alpha'' > \alpha$, the supremum $\alpha_T$ is not attained.

	Now assume for contradiction that there exists some instance $I'$ with $m$ centers where the $\alpha_m^*$-Droop core is empty for $k=1$. By \Cref{lem:full-reduction}, there exists some instance $I''$ with $T \leq m$ centers and $n \leq T$ agents such that $G_{\alpha_{m}^*}(I'')$ contains a directed Hamiltonian cycle.
	By the observation above, the existence of such a cycle at $\alpha_m^*$ implies that $\alpha_T > \alpha_m^*$, a contradiction.

	Finally, fix $\alpha < \alpha_m^*$. Then there exists $T \le m$ such that $\alpha < \alpha_{T}$. Start from an instance with $T$ centers such that $G_\alpha$ contains a directed cycle of length $T$. If $T=m$, we are done. Otherwise, add $m-T$ additional centers sufficiently far away from all agents so that each new center is strictly Pareto dominated by more than a factor of $\alpha$ by some center from the original $T$ centers. Then each new center is not a sink, while the original centers already lie on a directed cycle. Hence, the resulting instance has no sink, and therefore no $\alpha$-Droop core for $k=1$ exists.
\end{proof}

\subsection{A MILP to Search for a Lower Bound}\label{sec:MILP}

With the reduction from \Cref{lem:cycle-search-m} in hand, we will next show how to formulate a Mixed Integer Linear Program (MILP) that can search through all (reduced) clustering instances for a fixed $m$, and find the one with the largest possible value of $\alpha$ for which the $\alpha$-Droop core is empty for $k=1$.

\Cref{fig:LP-distance-search} shows such a MILP. It takes as input a fixed value of $\alpha$ and a value of $m$, and checks whether there exists a clustering instance with $m$ centers and $m$ weighted agents whose $\alpha$-deviation graph contains a directed cycle of length $m$. To find the maximum value of $\alpha$ for which an instance with $m$ centers and an empty $\alpha$-Droop core exists, one can repeatedly run the MILP while binary searching over values of $\alpha$.

The constraints of the program are divided into two parts. Constraints~\labelcref{eq:ILP-majority,eq:ILP-normalization,eq:ILP-nonneg-w,eq:ILP-binary} encode the assignment of agents to the deviation groups $S_0,\dots,S_{m-1}$. Specifically, each binary variable $x_{i,t}$ represents whether agent $i$ is (forced to be) in the group $S_t$ or not. These constraints ensure that the agents are assigned valid weights (constraints~\labelcref{eq:ILP-normalization,eq:ILP-nonneg-w,eq:ILP-binary})
and assigned to deviation groups (constraint~\labelcref{eq:ILP-majority})
such that each group $S_t$ contains a strict weighted majority of the agents (where the strictness is enforced by a pre-defined small constant $\epsilon$).

Constraints~\labelcref{eq:ILP-triangle,eq:ILP-deviation,eq:ILP-nonneg-d} take the found deviation structure, and solve for the actual metric distances that realizes this structure. We w.l.o.g. consider instances with $N \cap C = \emptyset$ as an agent-center $a \in N \cap C$ can be separated to $i_a \in N$ and $c_a \in C$ located in the same place. Then it suffices to encode the information of agent-center distances $(d(i,c))_{i \in N, c \in C}$, as the remaining distances can be given by the shortest path distance along the weighted complete bipartite graph.
Constraint~\labelcref{eq:ILP-triangle} ensures that the (shortest path) distances meet the triangle inequality, and hence form a metric space together with \labelcref{eq:ILP-nonneg-d}. Constraint~\labelcref{eq:ILP-deviation} ensures that for each agent $i \in N$ with $i \in S_t$, agent $i$ indeed wants to deviate from $c_{t-1}$ to $c_t$ by a factor of $\alpha$.

With this program, the utility of our reductions in \Cref{sec:reduce} becomes clear. In practice, one of the main bottlenecks in solving such an MILP is the number of integer variables in the program. Note that in \Cref{fig:LP-distance-search}, the only integer variables are the variables $(x_{i,t})_{i \in N, t \in [m]_0}$ (the variables representing the weights of the agents and the metric distances are all continuous). There is an $x$ variable for each agent $i \in N$ and each possible deviating group $S_t$. By \Cref{lem:m-agents-are-sufficient}, the number of agents is at most $m$.

Moreover, since the only deviation edges the program needs to consider are those of the directed Hamiltonian cycle, there are exactly $m$ edges, and hence only $m$ deviation sets to consider. Together, this bounds the number of integer variables by $m^2$, allowing the MILP to terminate quickly for small values of $m$.

\newcommand{\lpq}[3][17em]{
	\mathmakebox[#1][l]{#2}\quad #3
}

\begin{figure}[ht]
	\centering
\begin{subequations}\label{eq:LP-distance-search}
	\begin{tcolorbox}[colback=gray!8, colframe=gray!35, arc=5pt, boxrule=0.5pt,
		left=8pt, right=8pt, top=6pt, bottom=6pt]
		\begin{align}
			& \text{find feasible values for} \; (d(i,c))_{i \in N, c \in C},\,
			(w_i)_{i \in N},\, (x_{i,t})_{i \in N, t \in [m]_0} \notag \\[4pt]
			& \text{subject to} \notag \\
			& \lpcmt{Finds the optimal deviation structure} \notag \\[6pt]
			& \quad \lpq[20em]
			{\sum_{i \in N} w_i x_{i,t} \ge \tfrac{1}{2} + \epsilon,}
			{\forall\, t \in [m]_0,}
			\tag{1}\label{eq:ILP-majority} \\
			& \quad \sum_{i \in N} w_i = 1,
			\tag{2}\label{eq:ILP-normalization} \\
			& \quad \lpq[20em]
			{w_i \ge 0,}
			{\forall\, i \in N,}
			\tag{3}\label{eq:ILP-nonneg-w} \\
			& \quad \lpq[20em]
			{x_{i,t} \in \{0,1\},}
			{\forall\, i \in N,\; t \in [m]_0,}
			\tag{4}\label{eq:ILP-binary} \\[10pt]
			& \lpcmt{Finds the optimal distances that satisfy the deviation structure} \notag \\[6pt]
			& \quad \lpq[20em]
			{d(i,c) \le d(i,c') + d(i',c') + d(i',c),}
			{\forall\, i,i' \in N,\; c,c' \in C,}
			\tag{5}\label{eq:ILP-triangle} \\
			& \quad \lpq[20em]
			{d(i,c_{t-1}) \ge \alpha \cdot d(i,c_t) + \epsilon,}
			{\forall\, i \in N,\; t \in [m]_0 \text{ with } x_{i,t}=1,}
			\tag{6}\label{eq:ILP-deviation} \\
			& \quad \lpq[20em]
			{d(i,c) \ge 0,}
			{\forall\, i \in N,\; c \in C.}
			\tag{7}\label{eq:ILP-nonneg-d}
		\end{align}
	\end{tcolorbox}
\end{subequations}
\caption{MILP for finding a metric instance with $m$ centers where the $\alpha$-Droop core is empty for $k=1$, given fixed values of $\alpha$ and $m$. A feasible solution yields a clustering instance whose deviation graph $G_\alpha$ is a directed cycle of length $m$.}\label{fig:LP-distance-search}
\end{figure}

Note that two constraints in \Cref{fig:LP-distance-search} are stated in simplified form for readability, and each requires a standard ILP transformation in practice. Constraint~\eqref{eq:ILP-deviation} is conditional on $x_{i,t}=1$, and is enforced using the big-$M$ method: it is replaced by $d(i,c_{t-1}) \ge \alpha \cdot d(i,c_{t}) + \epsilon - M(1-x_{i,t})$ for all $i \in N$ and $t \in \{0,\ldots,m-1\}$, where $M$ is a sufficiently large constant. When $x_{i,t}=1$, the big-$M$ term vanishes and the constraint binds; when $x_{i,t}=0$, the constraint is trivially satisfied. Similarly, constraint~\eqref{eq:ILP-majority} contains the bilinear product $w_i\,x_{i,t}$. Since $x_{i,t}$ is binary and $w_i\in[0,1]$, this product can be linearized exactly by introducing an auxiliary variable $z_{i,t}\in[0,1]$ and adding the constraints $z_{i,t}\le w_i$, $z_{i,t}\le x_{i,t}$, and $z_{i,t}\ge w_i+x_{i,t}-1$, which together enforce $z_{i,t}=w_i x_{i,t}$. With both transformations applied, the program is a pure MILP.

\subsection{Optimal Bounds for Small Values of $m$}

Using the MILP outlined in \Cref{fig:LP-distance-search}, fixing a value of $m$, and binary searching over possible values of $\alpha$, we were able to determine $\alpha_m$ for $m \in \{3,4,5,6\}$: the boundary value of $\alpha$ such that, for some instance $I$, $G_\alpha (I)$ contains a directed cycle of length $m$.\footnote{The cases for $m \in \{1,2\}$ are trivial: $\alpha_m = 1$, as in either case there exists an agent $\alpha$-deviating to the same center.} These values are shown in \Cref{tab:best-alphas-cyclic}. Using these values in combination with \Cref{lem:cycle-search-m}, we can determine the minimum possible value $\alpha_m^*$ of $\alpha$ for which any instance with $m$ centers has a nonempty $\alpha$-Droop core for $k=1$. These bounds are shown in \Cref{tab:best-alphas-repeated}.

Note that, for three of the four values of $m$ we consider, their entries in \Cref{tab:best-alphas-cyclic} and \Cref{tab:best-alphas-repeated} are identical, meaning that the highest value of $\alpha$ one can achieve with a cycle of length $m$ is at least as high as the optimal one can achieve using a cycle of length $m'$ for any $m' < m$. However, this is interestingly not the case when $m=4$, where the optimal value that can be achieved by a $4$-cycle instance is less than $2$, the optimal value achievable by a $3$-cycle instance. It is unclear whether $m=4$ is the only exception to monotonicity of the optimal $\alpha$ for these cyclic instances.

\begin{table}[ht]
	\centering
	\begin{minipage}[t]{0.42\textwidth}
		\centering
		\begin{tabular}{cc}
			\toprule
			$T$ & $\alpha_T$ \\
			\midrule
			3 & 2 \\
			4 & 1.8393 \\
			5 & 2.0775 \\
			6 & 2.1019 \\
			\bottomrule
		\end{tabular}
		\caption{Tight values of $\alpha$ for cycle lengths
			 $T \in \{3, 4, 5, 6\}$.}
		\label{tab:best-alphas-cyclic}
	\end{minipage}
	\hfill
	\begin{minipage}[t]{0.42\textwidth}
		\centering
		\begin{tabular}{cc}
			\toprule
			$m$ & $\alpha^*_m$ \\
			\midrule
			3 & 2 \\
			4 & 2 \\
			5 & 2.0775 \\
			6 & 2.1019 \\
			\bottomrule
		\end{tabular}
		\caption{Tight $\alpha$-values for instances with $m$ centers for $m \in \{3, 4, 5, 6\}$.}
		\label{tab:best-alphas-repeated}
	\end{minipage}
\end{table}

Although the bounds in \Cref{tab:best-alphas-cyclic} were originally found by running the MILP from \Cref{fig:LP-distance-search}, we later analyzed the resulting metric spaces and manually derived the exact irrational values corresponding to the approximate values returned by the MILP. Below, rather than simply appealing to the outputs of the MILP as a certificate of correctness, we include formal proofs for the tight optimal $\alpha_m$ values for $m \in \{3,4,5,6\}$, which directly analyze the deviation graph, and may prove more useful to the reader than a simple correctness certificate produced by the MILP.

The style of these proofs will be very similar to \Cref{thm:general-case}, where we leveraged cycles that must necessarily appear in the deviation graphs of instances to rederive the best known upper-bound of $1 + \sqrt{2}$ for the general case. However, to derive these more nuanced bounds, we must use more complex analysis than \Cref{lem:deviation-step}. Below, we introduce two more ways we can bound distances along a deviation cycle.

\begin{lemma}\label{lem:cases-full}
	Let $I=(N,C,d,w)$ be a weighted instance and let $c_0c_1\cdots c_{T-1}c_0$ be a directed cycle in $G_\alpha(I)$ for some $\alpha > 1$, with deviation sets $S_t$ and overlap agents $i_t \in S_t \cap S_{t+1}$ for each $t \in [T]_0$. The following hold.
	\begin{enumerate}
		\item\label{item:full-3-consec} If $i_t = i_{t+1}$, then $d(i_{t+1}, c_{t+1}) < \frac{1}{\alpha} d(i_t, c_t)$.
		\item\label{item:full-2-jump-1} If $i_t \in S_{t+3}$, then $d(i_{t+2}, c_{t+2}) < \frac{1}{\alpha(\alpha - 1)} \left(d(i_{t+1},c_{t+1}) + d(i_t, c_t)\right)$.
	\end{enumerate}
\end{lemma}
\begin{proof}
	(\ref{item:full-3-consec}) This follows trivially.
	$$d(i_{t+1},c_{t+1}) = d(i_{t},c_{t+1}) < \frac{1}{\alpha} d(i_t, c_t).$$

	(\ref{item:full-2-jump-1}) We first have
	\begin{align*}
		\alpha^2 \cdot d(i_{t+2},c_{t+2}) &<  \alpha \cdot d(i_{t+2}, c_{t+1}) \\
		&\le  \alpha \cdot d(i_{t+2}, c_{t+3}) + \alpha \cdot d(i_{t}, c_{t+3}) + \alpha \cdot d(i_{t}, c_{t+1}) \\
		&< d(i_{t+2}, c_{t+2}) + d(i_{t}, c_{t+2}) + d(i_{t}, c_{t}), \\
	\end{align*}
	where the strict inequalities come from the definition of $S_t$ and the second inequality comes from the triangle inequality. Then, we have
	\begin{align*}
		d(i_{t},c_{t+2}) &\le  d(i_t, c_{t+1}) + d(i_{t+1}, c_{t+1}) + d(i_{t+1}, c_{t+2})\\
		&<\frac{1}{\alpha} \cdot d(i_{t}, c_{t}) + \left( 1 + \frac{1}{\alpha}\right) d(i_{t+1}, c_{t+1}),
	\end{align*}
	where the strict inequality comes from $S_t$ and the other from the triangle inequality.
	Hence rearranging terms, we get
	$$d(i_{t+2}, c_{t+2}) < \frac{1}{\alpha(\alpha - 1)} \left(d(i_{t+1},c_{t+1}) + d(i_t, c_t)\right).$$
\end{proof}

With \Cref{lem:cases-full}, we can proceed with full proofs for each considered $m$ value.

\begin{theorem}\label{thm:T=3}
	If $G_\alpha(I)$ for some instance $I$ contains a directed cycle of length 3, then $\alpha < 2$. Moreover, there exists an instance for which $G_\alpha$ contains a directed cycle of length 3 for every $\alpha < 2$.
\end{theorem}
\begin{proof}
	Let $c_0c_1c_2c_0$ be a directed cycle of length 3 contained in $G_\alpha(I)$ for some instance $I$. As in \Cref{thm:general-case}, let $S_t$ be the set that $\alpha$-deviates from $c_{t-1}$ to $c_t$, and let $i_t \in S_{t} \cap S_{t+1}$.

	Then for each $t \in [3]_0$, we have
	\begin{align*}
	\alpha \cdot d(i_{t},c_{t}) &<  d(i_{t}, c_{t-1}) \\
	&\le  d(i_{t}, c_{t+1}) + d(i_{t+1}, c_{t+1}) + d(i_{t+1}, c_{t-1}) \\
	&<  \frac{1}{\alpha} \cdot d(i_{t}, c_{t})  + \left(1 + \frac{1}{\alpha} \right) d(i_{t+1}, c_{t+1}),
	\end{align*}
	where the last inequality comes from the fact that $c_{t-1}=c_{t+2}$. Rearranging terms, we get
	$$ d(i_t, c_t) < \frac{1}{\alpha-1} \cdot d(i_{t+1},c_{t+1}) ,$$
	which implies $\alpha < 2$.

	To show tightness, consider $I = (\{i_0,i_1,i_2\}, \{c_0,c_1,c_2\}, d, w)$ where $w_{i_0}=w_{i_1}=w_{i_2}=1/3$ and distances $d$ as given in \Cref{fig:T3-instance} (a).

	\begin{figure}[h]
	    \centering
	    \begin{subfigure}[c]{0.28\linewidth}
	        \centering
	        $\begin{array}{c|ccc}
	            & c_0 & c_1 & c_2 \\
	            \hline
	            i_0 & 2 & 1 & 4 \\
	            i_1 & 4 & 2 & 1 \\
	            i_2 & 1 & 4 & 2
	        \end{array}$
	        \caption{Distance matrix $[d(i, c)]$ between each agent and center.}
	    \end{subfigure}
	    \hfill
	    \begin{subfigure}[c]{0.30\linewidth}
	        \centering
	        \begin{tikzpicture}[
	          cnode/.style={draw, circle, thick, minimum size=7mm, inner sep=1pt, fill=white},
	          >=Stealth, font=\small,
	        ]
	        \node[cnode] (c0) at (0,   0)   {$c_0$};
	        \node[cnode] (c1) at (3,   0)   {$c_1$};
	        \node[cnode] (c2) at (1.5, 2.6) {$c_2$};
	        \draw[->, thick] (c0) -- node[below] {$S_1$} (c1);
	        \draw[->, thick] (c1) -- node[right] {$S_2$} (c2);
	        \draw[->, thick] (c2) -- node[left]  {$S_0$} (c0);
	        \end{tikzpicture}
	        \caption{Deviation graph $G_\alpha$: the directed 3-cycle $c_0c_1c_2c_0$, with the deviating set $S_t$ labelling each edge.}
	    \end{subfigure}
	    \hfill
	    \begin{subfigure}[c]{0.33\linewidth}
	        \centering
	        $\begin{array}{cc|ccc}
	            w_i & & S_0 & S_1 & S_2 \\
	            \hline
	            \nicefrac{1}{3} & i_0 & 1 & 1 & 0 \\
	            \nicefrac{1}{3} & i_1 & 0 & 1 & 1 \\
	            \nicefrac{1}{3} & i_2 & 1 & 0 & 1 \\
	            \hline
	            & \text{sum} & \nicefrac{2}{3} & \nicefrac{2}{3} & \nicefrac{2}{3}
	        \end{array}$
	        \caption{Set membership: entry is $1$ iff agent $i_t$ belongs to deviating set $S_t$ for all $\alpha<\alpha_3$; bottom row gives the total weight of each set.}
	    \end{subfigure}
	    \caption{Tightness instance for $m = 3$: the distance matrix (a), the resulting deviation graph (b), and the deviating set memberships (c).}
	    \label{fig:T3-instance}
	\end{figure}
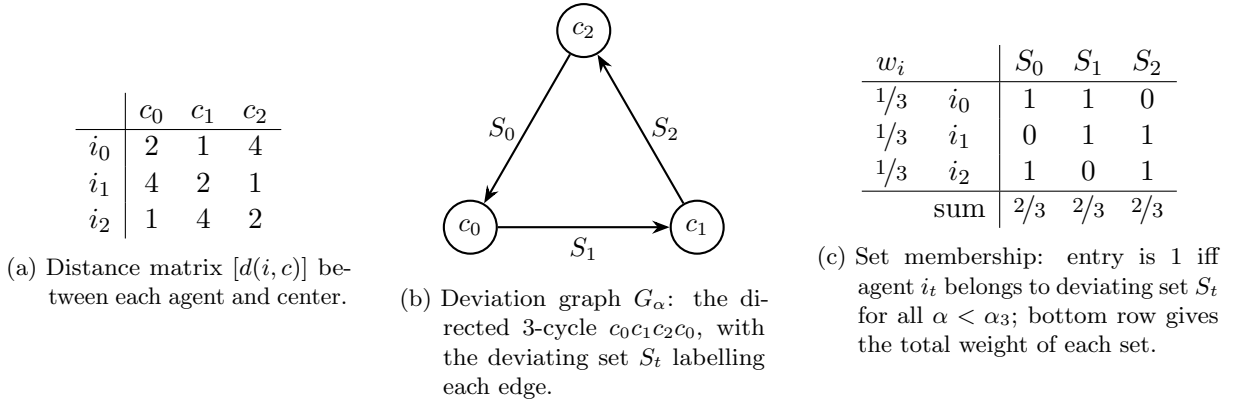
	Then, for each $c_t$ and $\alpha < 2$, agents $\{i_t, i_{t+1}\}$ $\alpha$-deviate to $c_{t+1}$, so $G_\alpha(I)$ contains the directed cycle $c_0c_1c_2c_0$.
\end{proof}

\begin{theorem}\label{thm:T=4}
	 Let $\alpha_4$ be the unique real root of $P_4(\alpha)=\alpha^3-\alpha^2 - \alpha - 1$.
	 If $G_\alpha(I)$ for some instance $I$ contains a directed cycle of length 4, then
	\[
	\alpha < \alpha_4 =
	\frac{1+\sqrt[3]{19+3\sqrt{33}}+\sqrt[3]{19-3\sqrt{33}}}{3}
	\approx 1.8393.
	\]
	 Moreover, there exists an instance for which $G_\alpha$ contains a directed cycle of length 4 for every $\alpha < \alpha_4$.
\end{theorem}
\begin{proof}
	Let $G_\alpha(I)$ for some instance $I$ contain a directed cycle $c_0c_1c_2c_3c_0$ of length 4. W.l.o.g. assume agents have uniform weights. Then since $|S_t| > \frac{n}{2}$, $\sum_{t=0}^3 |S_t| > 2n$ and hence there exists $i \in N$ that is contained in at least three of $S_0, S_1, S_2, S_3$. W.l.o.g. let $i \in S_1 \cap S_2 \cap S_3$, and let $i_1=i_2=i$. Take $i_3 \in S_3 \cap S_0$. Then,
	\begin{align*}
		d(i, c_0) &\le  d(i, c_{3}) + d(i_{3}, c_{3}) + d(i_{3}, c_{0}) \\
		&<  \frac{1}{\alpha^3} \cdot d(i, c_{0})  + \left(1 + \frac{1}{\alpha} \right) d(i_{3}, c_{3}).
	\end{align*}
	Also, we have
	$$d(i_3,c_3) < \frac{\alpha+1}{\alpha(\alpha - 1)} d(i_2, c_2) < \frac{\alpha + 1}{\alpha^3 (\alpha -1)} d(i, c_0),$$
	where the first inequality comes from \Cref{lem:deviation-step} and the latter comes from the fact that $i_2 = i$.  Hence, we have
	$$ \left( 1 - \frac{1}{\alpha^3} \right) d(i, c_0) < \frac{(\alpha+1)^2}{\alpha^4(\alpha-1)} \cdot d(i,c_{0}) .$$

	Dividing by $d(i,c_0)>0$ and multiplying through by $\alpha^4(\alpha-1)>0$, this rearranges to $\alpha(\alpha-1)^2(\alpha^2+\alpha+1) < (\alpha+1)^2$. The left-hand side minus the right-hand side equals $(\alpha^2+1)\,P_4(\alpha)$, as one can verify by expanding. Since $\alpha^2+1>0$, this gives $P_4(\alpha) < 0$, i.e.\ $\alpha < \alpha_4$.

	To show tightness, let $I = (\{i_0,i_1,i_2\}, \{c_0,c_1,c_2,c_3\}, d, w)$ where $w_{i_0}=w_{i_1}=w_{i_2}=1/3$ and distances $d$ as given in \Cref{fig:T4-instance} (a).

	\begin{figure}[h]
	    \centering
	    \begin{subfigure}[c]{0.33\linewidth}
	        \centering
	        $\begin{array}{c|cccc}
	            & c_0 & c_1 & c_2 & c_3 \\
	            \hline
	            i_0 & \alpha_4^3 & \alpha_4^2 & \alpha_4 & 1 \\
	            i_1 & \alpha_4 & 1 & \alpha_4^3 & \alpha_4^2 \\
	            i_2 & \alpha_4 & \alpha_4 \gamma & \gamma & \alpha_4^2
	        \end{array}$
	        \caption{Distance matrix $[d(i, c)]$ between each agent and center.}
	    \end{subfigure}
	    \hfill
	    \begin{subfigure}[c]{0.24\linewidth}
	        \centering
	        \begin{tikzpicture}[
	          cnode/.style={draw, circle, thick, minimum size=7mm, inner sep=1pt, fill=white},
	          >=Stealth, font=\small,
	        ]
	        \node[cnode] (c0) at (0.9, 0)   {$c_0$};
	        \node[cnode] (c1) at (1.8, 0.9) {$c_1$};
	        \node[cnode] (c2) at (0.9, 1.8) {$c_2$};
	        \node[cnode] (c3) at (0,   0.9) {$c_3$};
	        \draw[->, thick] (c0) -- node[below right, inner sep=2pt] {$S_1$} (c1);
	        \draw[->, thick] (c1) -- node[above right, inner sep=2pt] {$S_2$} (c2);
	        \draw[->, thick] (c2) -- node[above left,  inner sep=2pt] {$S_3$} (c3);
	        \draw[->, thick] (c3) -- node[below left,  inner sep=2pt] {$S_0$} (c0);
	        \end{tikzpicture}
	        \caption{Deviation graph $G_\alpha$: the directed 4-cycle $c_0c_1c_2c_3c_0$, with deviating set $S_t$ labelling each edge.}
	    \end{subfigure}
	    \hfill
	    \begin{subfigure}[c]{0.37\linewidth}
	        \centering
	        $\begin{array}{cc|cccc}
	            w_i & & S_0 & S_1 & S_2 & S_3 \\
	            \hline
	            \nicefrac{1}{3} & i_0 & 0 & 1 & 1 & 1 \\
	            \nicefrac{1}{3} & i_1 & 1 & 1 & 0 & 1 \\
	            \nicefrac{1}{3} & i_2 & 1 & 0 & 1 & 0 \\
	            \hline
	            & \text{sum} & \nicefrac{2}{3} & \nicefrac{2}{3} & \nicefrac{2}{3} & \nicefrac{2}{3}
	        \end{array}$
	        \caption{Set membership: entry is $1$ iff agent $i_t$ belongs to deviating set $S_t$ for all $\alpha<\alpha_4$; bottom row gives the total weight of each set.}
	    \end{subfigure}
	    \caption{Tightness instance for $m = 4$: the distance matrix (a), the resulting deviation graph (b), and the deviating set memberships (c). Here $\gamma = \alpha_4^3 - 2\alpha_4 = 2 + \frac{1}{\alpha_4}$.}
	    \label{fig:T4-instance}
	\end{figure}
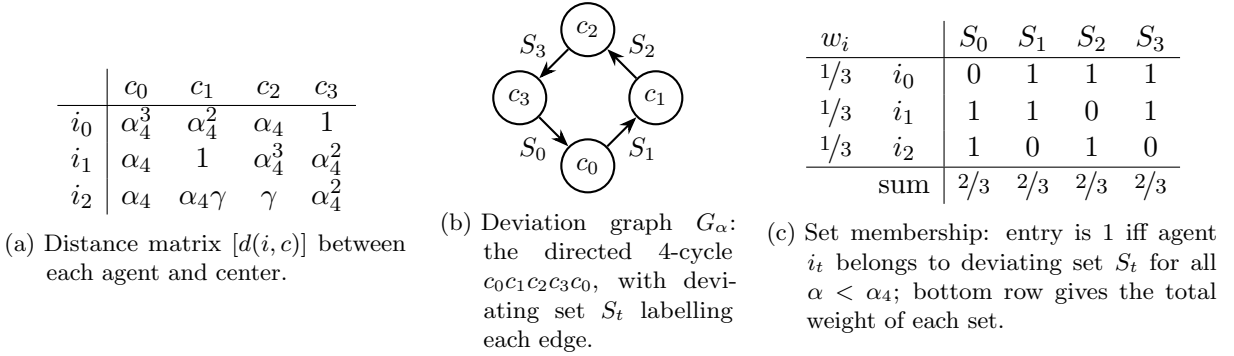
	Then, for each $c_t$ and $\alpha < \alpha_4$, two agents $\alpha$-deviate to $c_{t+1}$, so $G_\alpha(I)$ is the directed cycle $c_0c_1c_2c_3c_0$.
\end{proof}

\begin{theorem}\label{thm:T=5}
	Let $\alpha_5$ be the unique real root of $P_5(\alpha)=\alpha^5-3\alpha^4+3\alpha^3-3\alpha^2+3\alpha-3$.
	If $G_\alpha(I)$ for some instance $I$ contains a directed cycle of length $5$, then
	\[
		\alpha < \alpha_5 \approx 2.0775.
	\]
	 Moreover, there exists an instance for which $G_\alpha$ contains a directed cycle of length 5 for every
	$\alpha<\alpha_5$.

\end{theorem}
\begin{proof}
	Let $G_\alpha(I)$ for some instance $I$ contain a directed cycle $c_0c_1c_2c_3c_4c_0$ of length 5. W.l.o.g. assume agents have uniform weights.
	For each $t \in [5]_0$, we have $\sum_{\ell=0}^3|S_{t+ \ell}| > 2n$, hence there exists an agent $i$ in at least three of $S_{t}, S_{t+1}, S_{t+2}, S_{t+3}$.
	We divide into two cases.

	First, assume there exists some $t$ and $i$ with $i \in S_t \cap S_{t+1} \cap S_{t+2}$. By cyclic symmetry, assume that $i\in S_0\cap S_1\cap S_2$. Then consider the sets $S_{2}, S_{3}, S_{4}, S_{0}$. Since their total size is greater than $2n$, some agent $i'$ belongs to at least three of them.

	If $i' \in S_{2} \cap S_{3} \cap S_{4}$, then by taking $i_0=i_1=i$ and $i_2=i_3=i'$, applying \Cref{lem:deviation-step} and \Cref{lem:cases-full}~(\ref{item:full-3-consec}) yields
	$$d(i_0, c_0) < \left(\frac{1}{\alpha} \right)^2 \left(\frac{\alpha+1}{\alpha ( \alpha - 1)}  \right) ^3 d(i_0, c_0),$$
	which implies $\alpha < \overline{\alpha} \approx 1.986$.
	The case $i' \in S_{3} \cap S_{4} \cap S_{0}$ is symmetric.

	If $i' \in S_{2} \cap S_{3} \cap S_{0}$, applying \Cref{lem:cases-full}~(\ref{item:full-2-jump-1}) for $i_2=i'$, we have
	$$d(i_4, c_4) < \frac{1}{\alpha(\alpha - 1)} \left(d(i_{3},c_{3}) + d(i_2, c_2)\right) < \frac{\alpha^2 +1}{\alpha^2 ( \alpha - 1)^2} d(i_2, c_2).$$
	Taking $i_0=i_1=i$ and applying the above, we have
	$$d(i_0, c_0) < \left(\frac{1}{\alpha} \right) \left( \frac{\alpha^2 +1}{\alpha^2 ( \alpha - 1)^2}  \right) \left(\frac{\alpha+1}{\alpha ( \alpha - 1)}  \right) ^2 d(i_0, c_0),$$
	which implies $\alpha <  \overline{\alpha}' \approx 2.0700$.

	Finally if $i' \in S_{2} \cap S_{4} \cap S_{0}$, we show that $\alpha < 2$ must hold.
	Assume, towards a contradiction, that $\alpha \ge 2$. Consider the following weighted sum:
	\begin{align*}
		\Phi
		={}&
		1\bigl(d(i,c_2)+d(i_2,c_2)+d(i_2,c_3)-d(i,c_3)\bigr)
		\\
		&+8\bigl(d(i,c_2)+d(i',c_2)+d(i',c_4)-d(i,c_4)\bigr)
		\\
		&+12\bigl(d(i',c_2)+d(i,c_2)+d(i,c_1)-d(i',c_1)\bigr)
		\\
		&+1\bigl(d(i',c_2)+d(i,c_2)+d(i,c_3)-d(i',c_3)\bigr)
		\\
		&+3\bigl(d(i',c_2)+d(i_2,c_2)+d(i_2,c_3)-d(i',c_3)\bigr)
		\\
		&+6\bigl(d(i_2,c_2)+d(i,c_2)+d(i,c_1)-d(i_2,c_1)\bigr)
		\\
		&+8\bigl(d(i,c_4)-2d(i,c_0)\bigr)
		\\
		&+16\bigl(d(i,c_0)-2d(i,c_1)\bigr)
		\\
		&+14\bigl(d(i,c_1)-2d(i,c_2)\bigr)
		\\
		&+12\bigl(d(i',c_1)-2d(i',c_2)\bigr)
		\\
		&+4\bigl(d(i',c_3)-2d(i',c_4)\bigr)
		\\
		&+6\bigl(d(i_2,c_1)-2d(i_2,c_2)\bigr)
		\\
		&+2\bigl(d(i_2,c_2)-2d(i_2,c_3)\bigr).
	\end{align*}
	The first six parenthesized terms are nonnegative by triangle inequality, and the last seven parenthesized terms are strictly positive as $\alpha \ge 2$. Therefore, $\Phi>0$.
	On the other hand, after distributing the coefficients and grouping equal distance terms, all terms cancel. Hence $\Phi=0$, a contradiction.

	Now suppose no agent belongs to three consecutive deviation sets. For each $t\in[5]_0$, define
	\[
		S_t' = S_t \cap S_{t+1} \cap S_{t+3},
	\]
	and let $J = \{t \in [5]_0 : S_t' \neq \emptyset\}$. We claim that for every $a \in [5]_0$, at least one of $a$ or $a+2$ lies in $J$. Fix any such pair $\{a, a+2\}$. The four consecutive sets $S_a, S_{a+1}, S_{a+2}, S_{a+3}$ each have total weight $> \tfrac{1}{2}$, so their combined weight exceeds $2 = 2\sum_i w_i$. Hence some agent $i^*$ lies in at least three of them. Since no agent is in three consecutive sets, $i^*$ cannot be in $\{S_a, S_{a+1}, S_{a+2}\}$ or $\{S_{a+1}, S_{a+2}, S_{a+3}\}$. The only remaining non-three-consecutive triples within $\{S_a, S_{a+1}, S_{a+2}, S_{a+3}\}$ are $\{S_a, S_{a+1}, S_{a+3}\}$ and $\{S_a, S_{a+2}, S_{a+3}\}$. In the first case $i^* \in S_a'$, so $a \in J$; in the second $i^* \in S_{a+2}'$, so $a+2 \in J$. Thus the pair $\{a,a+2\}$ is hit.

	Since each index lies in exactly two of the five pairs $\{a, a+2\}$, a set of two indices covers at most four such pairs; hence $|J| \ge 3$. Moreover, $J$ must contain three consecutive elements: otherwise $|J| = 3$ with no three consecutive, forcing $J = \{t, t+1, t+3\}$ for some $t$ (these are the only size-$3$ subsets of $[5]_0$ with no three consecutive elements). But then neither $t+2$ nor $t+4$ is in $J$, so the pair $\{t+2, t+4\}$ is not covered, a contradiction.

	Hence, by cyclic symmetry, assume $0,1,2 \in J$ and pick $i_0 \in S_0'$, $i_1 \in S_1'$, $i_2 \in S_2'$. Again let $i_3 \in S_3 \cap S_4$ and $i_4 \in S_4 \cap S_0$. Suppose, toward a contradiction, that $G_\alpha$ contains a directed cycle of length $5$ for some $\alpha \ge \alpha_5$. Then, $G_{\alpha_5}$ must also contain the same cycle. It is hence sufficient to show that $G_{\alpha_5}$ contains no directed $5$-cycle.

	For notational convenience, below we set $\alpha = \alpha_5$.
	Define the following $3$-path residuals:
	\begin{align*}
		T_1&=d(i_0,c_1)+d(i_1,c_1)+d(i_1,c_2)-d(i_0,c_2),\\
		T_2&=d(i_0,c_3)+d(i_3,c_3)+d(i_3,c_4)-d(i_0,c_4),\\
		T_3&=d(i_1,c_4)+d(i_4,c_4)+d(i_4,c_0)-d(i_1,c_0),\\
		T_4&=d(i_1,c_2)+d(i_2,c_2)+d(i_2,c_3)-d(i_1,c_3),\\
		T_5&=d(i_2,c_3)+d(i_0,c_3)+d(i_0,c_1)-d(i_2,c_1),\\
		T_6&=d(i_2,c_3)+d(i_3,c_3)+d(i_3,c_4)-d(i_2,c_4),\\
		T_7&=d(i_3,c_4)+d(i_1,c_4)+d(i_1,c_2)-d(i_3,c_2),\\
		T_8&=d(i_4,c_0)+d(i_2,c_0)+d(i_2,c_3)-d(i_4,c_3).
	\end{align*}
	Then each $T_j\ge 0$ by the triangle inequality.

	Next, define the deviation residuals
	\begin{align*}
		\Delta_1&=d(i_0,c_4)-\alpha d(i_0,c_0),\\
		\Delta_2&=d(i_0,c_0)-\alpha d(i_0,c_1),\\
		\Delta_3&=d(i_0,c_2)-\alpha d(i_0,c_3),\\
		\Delta_4&=d(i_1,c_0)-\alpha d(i_1,c_1),\\
		\Delta_5&=d(i_1,c_1)-\alpha d(i_1,c_2),\\
		\Delta_6&=d(i_1,c_3)-\alpha d(i_1,c_4),\\
		\Delta_7&=d(i_2,c_1)-\alpha d(i_2,c_2),\\
		\Delta_8&=d(i_2,c_2)-\alpha d(i_2,c_3),\\
		\Delta_9&=d(i_2,c_4)-\alpha d(i_2,c_0),\\
		\Delta_{10}&=d(i_3,c_2)-\alpha d(i_3,c_3),\\
		\Delta_{11}&=d(i_3,c_3)-\alpha d(i_3,c_4),\\
		\Delta_{12}&=d(i_4,c_3)-\alpha d(i_4,c_4),\\
		\Delta_{13}&=d(i_4,c_4)-\alpha d(i_4,c_0).
	\end{align*}
	By the choice of the indices $i_t$, each $\Delta_j>0$.

	Define multipliers $\lambda_1,\ldots,\lambda_8$ and $\mu$ by
	\begin{align*}
		\lambda_1&=8\alpha^4-8\alpha^3+8\alpha^2-5\alpha+11,\\
		\lambda_2&=5\alpha^4-6\alpha^3+7\alpha^2-6\alpha+8,\\
		\lambda_3&=9\alpha^4-8\alpha^3+9\alpha^2-7\alpha+15,\\
		\lambda_4&=8\alpha^4-8\alpha^3+11\alpha^2-8\alpha+13,\\
		\lambda_5&=11\alpha^4-10\alpha^3+12\alpha^2-7\alpha+16,\\
		\lambda_6&=4\alpha^4-3\alpha^3+3\alpha^2-2\alpha+4,\\
		\lambda_7&=7\alpha^4-5\alpha^3+7\alpha^2-4\alpha+9,\\
		\lambda_8&=9\alpha^4-9\alpha^3+10\alpha^2-8\alpha+12,\\
		\mu&=(3\alpha^2-5\alpha+4)(5\alpha^2+4\alpha+5).
	\end{align*}
	All these multipliers are strictly positive for every real $\alpha$.

	Now consider the weighted sum
	\begin{align*}
		\Phi={}&
		\lambda_1T_1+\lambda_2T_2+\lambda_3T_3+\lambda_4T_4
		+\lambda_5T_5+\lambda_6T_6+\lambda_7T_7+\lambda_8T_8\\
		&+\lambda_2\Delta_1+\lambda_3\Delta_2+\lambda_1\Delta_3
		+\lambda_3\Delta_4+\lambda_5\Delta_5+\lambda_4\Delta_6+\lambda_5\Delta_7\\
		&+\mu\Delta_8+\lambda_6\Delta_9
		+\lambda_7\Delta_{10}+\lambda_7\Delta_{11}
		+\lambda_8\Delta_{12}+\lambda_8\Delta_{13}.
	\end{align*}

	Since $T_j\ge 0$, $\Delta_j>0$, and every multiplier is strictly positive, we have $\Phi>0$.
	On the other hand, expanding $\Phi$ as a linear combination of the variables $d(i_r,c_s)$, all coefficients cancel. These cancellations are exactly the following identities:
	\begin{align*}
		\lambda_3&=\alpha\lambda_2,\\
		\lambda_1+\lambda_5&=\alpha\lambda_3,\\
		\lambda_2+\lambda_5&=\alpha\lambda_1,\\
		\lambda_1+\lambda_4+\lambda_7&=\alpha\lambda_5,\\
		\lambda_3+\lambda_7&=\alpha\lambda_4,\\
		\lambda_8&=\alpha\lambda_6,\\
		\lambda_4+\mu&=\alpha\lambda_5,\\
		\lambda_4+\lambda_5+\lambda_6+\lambda_8&=\alpha\mu,\\
		\lambda_2+\lambda_6+\lambda_7&=\alpha\lambda_7,\\
		\lambda_3+\lambda_8&=\alpha\lambda_8.
	\end{align*}
	Each identity follows from
	$$
	P_5(\alpha)=\alpha^5-3\alpha^4+3\alpha^3-3\alpha^2+3\alpha-3=0.
	$$
	Therefore $\Phi=0$, contradicting $\Phi > 0$.

	To show tightness, let $I = (\{i_0,i_1,i_2,i_3,i_4\}, \{c_0,c_1,c_2,c_3,c_4\}, d, w)$ where $w_{i_0}=w_{i_2}=2/9$, $w_{i_1}=1/3$, $w_{i_3}=w_{i_4}=1/9$ (equivalently, agents $i_0$ and $i_2$ have $2$ copies each, agent $i_1$ has $3$ copies, and agents $i_3$ and $i_4$ have $1$ copy each at the same location), and distances $d$ as given in \Cref{fig:T5-instance} (a).

	\begin{figure}[h]
	    \centering
	    \begin{subfigure}[c]{0.77\linewidth}
	        \centering
	        \scalebox{0.9}{$\begin{array}{c|ccccc}
	            & c_0 & c_1 & c_2 & c_3 & c_4\\
	            \hline
	            i_0 & \alpha_5^2 \delta^4 & \alpha_5 \delta^4 & \delta^4 & 2 \alpha_5 \delta^3& 2 \delta^3 \\
	            i_1 & 2 \delta^2 &  \alpha_5^2 \delta^3 & \alpha_5 \delta^3 & \delta^3 & 2 \alpha_5 \delta^2 \\
	            i_2 & 2\alpha_5 \delta& 2 \delta & \alpha_5^2 \delta^2 & \alpha_5 \delta^2 & \delta^2 \\
	            i_3 & \delta & 2 \alpha_5 & \alpha_5^2 \delta^4 - \delta^4 - \delta&  \alpha_5^2 \delta & \alpha_5 \delta\\
	            i_4 & \alpha_5 & 1 &  \alpha_5^2 \delta^3 - \delta^3 + 1 &  \alpha_5^2 \delta^3 - \delta^3 - 1 & \alpha_5^2 \\
	        \end{array}$}
	        \caption{Distance matrix $[d(i, c)]$ between each agent and center.}
	    \end{subfigure}
	    \hfill
	    \begin{subfigure}[c]{0.4\linewidth}
	        \centering
	        \begin{tikzpicture}[
	          cnode/.style={draw, circle, thick, minimum size=7mm, inner sep=1pt, fill=white},
	          >=Stealth, font=\small,
	        ]
	        \node[cnode] (c0) at (1.5,  2.4)  {$c_0$};
	        \node[cnode] (c1) at (2.64, 1.57) {$c_1$};
	        \node[cnode] (c2) at (2.21, 0.23) {$c_2$};
	        \node[cnode] (c3) at (0.79, 0.23) {$c_3$};
	        \node[cnode] (c4) at (0.36, 1.57) {$c_4$};
	        \draw[->, thick] (c0) -- node[above right, inner sep=1pt] {$S_1$} (c1);
	        \draw[->, thick] (c1) -- node[right,        inner sep=1pt] {$S_2$} (c2);
	        \draw[->, thick] (c2) -- node[below=1pt, inner sep=1pt] {$S_3$} (c3);
	        \draw[->, thick] (c3) -- node[left,         inner sep=1pt] {$S_4$} (c4);
	        \draw[->, thick] (c4) -- node[above left,   inner sep=1pt] {$S_0$} (c0);
	        \end{tikzpicture}
	        \caption{Deviation graph $G_\alpha$: the directed 5-cycle $c_0 \cdots c_4c_0$, with deviating set $S_t$ labelling each edge.}
	    \end{subfigure}
	    \hfill
	    \begin{subfigure}[c]{0.5\linewidth}
	        \centering
	        \scalebox{0.9}{
	        $\begin{array}{cc|ccccc}
	            w_i & & S_0 & S_1 & S_2 & S_3 & S_4 \\
	            \hline
	            \nicefrac{2}{9} & i_0 & 0 & 1 & 1 & 0 & 1 \\
	            \nicefrac{1}{3} & i_1 & 1 & 0 & 1 & 1 & 0 \\
	            \nicefrac{2}{9} & i_2 & 0 & 1 & 0 & 1 & 1 \\
	            \nicefrac{1}{9} & i_3 & 1 & 0 & 0 & 0 & 1 \\
	            \nicefrac{1}{9} & i_4 & 1 & 1 & 0 & 0 & 0 \\
	            \hline
	            & \text{sum} & \nicefrac{5}{9} & \nicefrac{5}{9} & \nicefrac{5}{9} & \nicefrac{5}{9} & \nicefrac{5}{9}
	        \end{array}$}
	        \caption{Set membership: entry is $1$ iff agent $i_t$ belongs to deviating set $S_t$ for all $\alpha<\alpha_5$; bottom row gives the total weight of each set.}
	    \end{subfigure}
	    \caption{Tightness instance for $m = 5$: the distance matrix (a), the resulting deviation graph (b), and the deviating set memberships (c). Here $\delta = \alpha_5 - 1$.}
	    \label{fig:T5-instance}
	\end{figure}
	Then, for every $c_t$ and every $\alpha < \alpha_5$, a strict majority of agents $\alpha$-deviate to $c_{t+1}$, so $G_\alpha(I)$ contains the directed cycle $c_0c_1c_2c_3c_4c_0$.
\end{proof}

The computations in the proof of \Cref{thm:T=5} can be verified with symbolic computations in $\mathbb{Q}[\alpha]/p(\alpha)$ by the accompanying code in the GitHub repository.\footnote{\url{https://github.com/evamichelle30/ProportionalClusteringLowerBounds}}

\begin{restatable}{theorem}{Tsix}\label{thm:T=6}
	Let $\alpha_6$ be the unique real root of $P_6(\alpha)=\alpha^5-2\alpha^4+\alpha^2-4\alpha+2$ in $[2,\infty)$.
	If $G_\alpha(I)$ for some instance $I$ contains a directed cycle of length 6, then
	\[
	\alpha < \alpha_6 \approx 2.1019.
	\]
	Moreover, there exists an instance for which $G_\alpha$ contains a directed cycle of length 6 for every $\alpha < \alpha_6$.
\end{restatable}
Since the proof is similar to previous theorems, we defer it to the appendix~(\Cref{sec:pfT=6}).

\subsection{The Best Known Lower Bound}

In the previous section, we determined the exact optimal values of $\alpha$ for several values of $m$. In the corresponding proofs, we provided, for each such value of $m$, a metric instance attaining this value of $\alpha$.

Although we were able to provide explicit proofs for each of these values, they were originally found by running the MILP formulated in \Cref{fig:LP-distance-search}. Beyond $m=6$, the MILP became too slow to continue deriving further $\alpha$ values.

As noted above, the only binary variables in \Cref{fig:LP-distance-search} are the variables $x_{i,t}$, which indicate whether agent $i$ belongs to deviation set $S_t$. If the deviation structure is fixed in advance, then these variables can instead be treated as constants. The resulting feasibility problem is then a linear program with only continuous variables, making it significantly more efficient to solve.

\begin{figure}[ht]
	\centering
	\begin{subequations}\label{eq:LP-distance-search-fast}
		\begin{tcolorbox}[colback=gray!8, colframe=gray!35, arc=5pt,
			boxrule=0.5pt, left=8pt, right=8pt, top=0pt, bottom=8pt]
			\begin{align}
				& \text{find feasible values for} \; (d(i,c))_{i \in N, c \in C} \notag \\[4pt]
				& \text{subject to} \quad \lpcmt{Constraints on distances} \notag \\[6pt]
				& \quad \lpq[20em]
				{d(i,c) \le d(i,c') + d(i',c') + d(i',c),}
				{\forall\, i,i' \in N,\; c,c' \in C,}
				\tag{1}\label{eq:LP-fast-triangle} \\
				& \quad \lpq[20em]
				{d(i,c_{t-1}) \ge \alpha \cdot d(i,c_t) + \epsilon,}
				{\forall\, i \in N,\; t \in [m]_0 \text{ with } x_{i,t}=1,}
				\tag{2}\label{eq:LP-fast-deviation} \\
				& \quad \lpq[20em]
				{d(i,c) \ge 0,}
				{\forall\, i \in N,\; c \in C.}
				\tag{3}\label{eq:LP-fast-nonneg-d}
			\end{align}
		\end{tcolorbox}
	\end{subequations}
	\caption{LP for finding a metric distance function $d$ satisfying the triangle inequality that witnesses an $\alpha$ lower bound, given a fixed deviation structure encoded by the $x_{i,t}$ constants.}
	\label{fig:LP-distance-search-fast}
\end{figure}

We present this faster program formally in \Cref{fig:LP-distance-search-fast}. If we can manually find a good assignment of agents to deviation sets, then we can use that as input to the linear program in \Cref{fig:LP-distance-search-fast} to check whether there exists a metric space that realizes this deviation structure and has an empty $\alpha$-Droop core for some fixed $\alpha$ and $k=1$. Unlike the MILP from \Cref{fig:LP-distance-search}, the metric space obtained by binary searching over values of $\alpha$ with this program is not guaranteed to be optimal for a fixed value of $m$. Indeed, there may exist another deviation structure that yields an empty Droop core for a larger value of $\alpha$. Nevertheless, this program is useful for finding new lower-bound instances. In particular, we found the new best lower bound mentioned in \Cref{thm:T=37} by first finding such a deviation structure by hand, using various heuristics.

In the appendix~(\Cref{app:metric37-certificate}), we provide the exact deviation structure
of an instance with $37$ centers and agents whose $2.1508$-Droop core is empty for $k=1$. This includes the sets $S_t$ to which each agent is assigned, as well as the corresponding weights that ensure that each deviation set contains a strict majority of the agents. In the accompanying \href{https://github.com/evamichelle30/ProportionalClusteringLowerBounds}{GitHub repository} for this paper, we provide an implementation that uses the LP from \Cref{fig:LP-distance-search-fast} to binary search over values of $\alpha$. The deviation structure from \Cref{app:metric37-certificate} can be given as input to this program to produce the metric space that yields this lower bound.

\section{Conclusion}

We improved the best known lower bound for proportionally fair clustering in general metric spaces from $2$ to $2.1508$. Our approach relies on the connection between proportional clustering and the $\beta$-plurality problem, put forward by \citet{kellerhals2026proportional}. By focusing on the Droop core for $k=1$, we obtain lower-bound instances that transfer to the standard Hare core.
Beyond the main lower bound, we analyzed small instances and obtained tight thresholds for deviation cycles of length $T\le 6$. The six-center construction (in the appendix) already gives an explicit and interpretable lower bound of approximately $2.1019$. Technically, we reduced the search space by showing that it suffices to consider directed cycles in the deviation graph and weighted instances with at most $m$ agents for $m$ centers. Combined with linear programming feasibility checks and local search over deviation structures, these reductions led to the $37$-center instance establishing the lower bound $2.1508$.

Several questions remain open. Most importantly, the exact value of the optimal core approximation in general metric spaces is still unknown. The current gap is now between $2.1508$ and $1+\sqrt{2}$. Closing this gap may require a stronger structural analysis of deviation cycles, new upper-bound techniques improving on Greedy Capture, or more exhaustive computational searches. It would also be interesting to determine whether similar lower-bound techniques apply to related proportionality notions in multiwinner voting, participatory budgeting, and non-centroid clustering.

\bibliography{refs}

\appendix

\section{Core equivalence for different quota}\label{sec:ABC}
In this section, we prove a slight generalization of \Cref{thm:core-equivalence} in terms of quota.
As stated in the main text, the proof of \Cref{thm:core-equivalence} is not specific to metric spaces. It only uses the property that we can make copies of agents and candidates such that, in each copy, an agent derives her loss or utility solely from the selected candidates in that same copy. This can be done, for instance, in ABC voting by setting each voter’s preferences so that they only approve candidates in the same copy.
To illustrate how \Cref{thm:core-equivalence} can be translated to these settings, we state it here for ABC voting.

\begin{theorem}
	Let $g$ be a nonnegative function such that $g(n,k)=o(1)$ as $n+k \to \infty$. Then the $\alpha$-Droop core is always nonempty if and only if the $\alpha$-core with quota $q=(1+g(n,k))\frac{n}{k}$ is always nonempty.
\end{theorem}

\begin{proof}Since $g$ is nonnegative, we have
	$q(n,k)=\left\lceil (1+g(n,k))\frac{n}{k}\right\rceil\ge \left\lceil \frac{n}{k}\right\rceil$, and $ \lceil \frac{n}{k}\rceil \ge \lfloor \frac{n}{k+1}\rfloor+1$.
	Hence every coalition that is blocking with respect to quota $q(n,k)$ is also
	feasible under the Droop quota. Therefore, nonemptiness of the $\alpha$-Droop
	core implies nonemptiness of the $\alpha$-core with quota $q$.

	Conversely, suppose that there exists an ABC voting instance $(N,C,A,k)$ whose
	$\alpha$-Droop core is empty. Let $n=|N|$, and set
	$\ell=\left\lfloor \frac{n}{k+1}\right\rfloor+1$.
	We show that there exists another instance whose $\alpha$-core with quota $q$
	is empty.

	For an integer $t$, set $k'=t(k+1)-1$. Consider an instance with $tn$ voters and committee size $k'$. Its quota is
	$q(tn,k')=(1+g(tn,k'))\frac{tn}{k'}$.
	Equivalently, define
	$f_t=\frac{k'g(tn,k')}{1+g(tn,k')}$.
	Then
	$(1+g(tn,k'))\frac{tn}{k'}=\frac{tn}{k'-f_t}$.
	Since $g(tn,k')=g(tn,t(k+1)-1)=o_t(1)$ and $k'=t(k+1)-1$, we have $f_t=o(t)$. Therefore,
	$$
	\frac{tn}{k'-f_t}
	=
	\frac{tn}{t(k+1)-1-f_t}
	=
	\frac{n}{k+1-\frac{1+f_t}{t}}
	\longrightarrow
	\frac{n}{k+1}
	$$
	as $t\to\infty$.
	Since $n/(k+1)<\ell$ and $\ell$ is an integer, there exists a sufficiently large
	integer $t$ such that $q(tn,k') < \ell$.

	Fix such $t$ and construct an instance $tA$ consisting of $t$ disjoint copies of the original
	instance. Thus we have voter sets $N_1,\ldots,N_t$ and candidate sets
	$C_1,\ldots,C_t$, and each voter $i\in N_j$ approves only candidates in $C_j$,
	according to the same approval set as in the original instance.

	Every committee $W$ of size $k'$ selects at most $k$ candidates from at least
	one copy $C_j$. Otherwise, each copy would contain at least $k+1$ selected
	candidates, requiring at least $t(k+1)>k'$ candidates in total. Since the original instance has an empty $\alpha$-Droop core, the resulting committee is blocked
	by some sets $S_j\subseteq N_j$ and $T_j\subseteq C_j$ satisfying $|S_j|\ge \ell |T_j|$ and
	\[
	|A_i\cap T_j|
	>
	\alpha |A_i\cap (W\cap C_j)|
	=
	\alpha |A_i\cap W|
	\]
	for all $i\in S_j$, where the identity follows from $A_i\subseteq C_j$.

	By the choice of $t$, the quota in the copied instance satisfies
	$q(tn,k')\le \ell$.
	Thus
	\[
	|S_j|\ge \ell |T_j|\ge q(tn,k')|T_j|.
	\]
	Therefore $S_j$ $\alpha$-blocks $W$ with respect to quota $q$. Since this holds
	for every committee $W$ of size $k'$, the $\alpha$-core with quota $q$ is empty
	in the copied instance.
\end{proof}

One possible choice is $q=\frac{n}{k-c}$, where $0<c<k$ is a constant.
Under this quota, a coalition can block with a subcommittee $T$ only if $|T|\le k-c$.

\section{A lower bound of 2 for $N=C$}\label{sec:N=C}

In this section, we prove a lower bound of 2 for the $\alpha$-core for the special case of $N=C$.
In this case, the lower bound for the $\alpha$-core was known to be $1.5$~\citep{ChenEtal2019}, with no progress since. We recently learned that a Bachelor's thesis (although in German) improved this lower bound to $2$ using a rather complicated instance~\citep{Trinh2024}.
Here, we provide the same bound of 2 using our framework of transferring a $\beta$-plurality instance to proportional clustering. We believe this approach gives a cleaner result.
\begin{theorem}
	For each $\alpha < 2$, there exists an instance with $N=C$ such that the $\alpha$-core is empty.
\end{theorem}
\begin{proof}
	We build the proof from the instance by \citet[Theorem 2]{filtser2020plurality}, which proves the lower bound of 2 in the $\beta$-plurality problem.

	The instance considers the metric space $(S^1, d)$ where $S^1$ is the unit sphere and $d$ is the distance along the arcs. By identifying $S^1 = [0,1)$ and scaling distances, they then let $C=S^1$ and $N=\{0, \frac{1}{3}, \frac{2}{3}\}$. Then they show that the $\alpha$-Droop core for $k=1$ is empty for $\alpha<2$.

	We discretize this instance to make it applicable for our case.
	For a positive integer $r$, consider the weighted instance $I_r=(N_r, C_r, d|_{N_r \times N_r})$, where $N_r=C_r=\{\frac{s}{6r}\}_{s=0}^{6r-1} \subseteq S^1$. We assign weights to the agents; equivalently, this can be done by copying voters according to the weights. The three agents $0, \frac{1}{3}, \frac{2}{3}$ each receive a weight of $\frac{3}{10}$, and the remaining weight $\frac{1}{10}$ is distributed uniformly to the remaining agents. Then any two of the heavily weighted agents constitute more than half of the total weight.

	We then proceed analogously to the original instance. Assume a center $c$ at $\frac{s}{6r}$ is selected, where w.l.o.g. $s \in \{0, 1, \ldots, r\}$ by symmetry. Then consider the center $c'$ at $\frac{1}{2} - \frac{1}{6r} \lceil \frac{s}{2} \rceil$. Then the two agents at $\frac{1}{3}$ and $\frac{2}{3}$ would $\alpha$-deviate from $c$ to $c'$ for every $\alpha < 2 - \frac{1}{2r}$. Hence, by taking $r \to \infty$, for $\alpha < 2$, in the special case of $N=C$, there exists an instance where the $\alpha$-Droop core for $k=1$ is empty. The theorem then follows by \Cref{thm:core-equivalence}.
\end{proof}

\section{Proof of \Cref{thm:T=6}}\label{sec:pfT=6}
In this section, we prove \Cref{thm:T=6}.
The proof again proceeds similarly to \Cref{thm:T=5}; it exhausts cases that have supporting weights, and proves a bound for $\alpha$ in each case using either \Cref{lem:deviation-step}, \Cref{lem:cases-full}, or Farkas certificates. Regarding parts that we use Farkas certificates, we separate each case with a lemma that appears at the end of the section. Each certificate can also be verified by the accompanying code in the GitHub repository, which uses symbolic computations in $\mathbb{Q}[\alpha]/p(\alpha)$.\footnote{\url{https://github.com/evamichelle30/ProportionalClusteringLowerBounds}}

\Tsix*
\begin{proof}
	Let $G_\alpha = c_0c_1c_2c_3c_4c_5c_0$ be a directed cycle of length 6.
	Since $\sum_{t=0}^5|S_{t}| > 3n$, there exists an agent $i$ that belongs to at least four of $S_{0}, \ldots, S_{5}$.

	If there exists an agent that resides in four consecutive $S_t$'s, say $S_0, S_1, S_2, S_3$, then by taking $i_0=i_1=i_2=i$ and applying \Cref{lem:deviation-step} and \Cref{lem:cases-full}~(\ref{item:full-3-consec}) yields
	$$d(i_0, c_0) < \left(\frac{1}{\alpha} \right)^2 \left(\frac{\alpha+1}{\alpha ( \alpha - 1)}  \right) ^4 d(i_0, c_0),$$
	which implies $\alpha < \overline{\alpha} \approx 2.0424$.

	Now assume no four consecutive $S_t$'s contain a common agent. If there exists an agent that resides in $S_t \cap S_{t+1} \cap S_{t+2} \cap S_{t+4}$ for some $t$, then we have $\alpha < 2.012$ by \Cref{lem:0124}.

	Now assume that no agent that is contained in four $S_t$'s is in three consecutive $S_t$'s.
	Then w.l.o.g. $i \in S_0 \cap S_1 \cap S_3 \cap S_4$.
	If there exists another four-common agent $i' \in S_1 \cap S_2 \cap S_4 \cap S_5$, then $\alpha < 2$ by \Cref{lem:0134-1245}.
	The case is symmetric for $S_2 \cap S_3 \cap S_5 \cap S_0$.
Hence, we can assume no other four $S_t$'s contain a common agent.

Then, for supporting weights to exist, we claim that there must exist an agent $i' \in S_t \cap S_2 \cap S_5$ for some $t \in \{0,1,3,4\}$.
This is because if such agents do not exist, then an agent $i^* \in S_2 \cap S_5$ must not be contained in any of $S_0, S_1, S_3, S_4$.
Let $p_0 = p_1 = p_3 = p_4 = \frac{1}{4}$ and $p_2=p_5 = \frac{1}{2}$. Then, for any agent $i$, if $i \in \bigcap_{t \in J} S_t$ for some allowed $J \subseteq [6]_0$, we have $\sum_{t \in J} p_t \le 1$.
Hence, for any weight profile $w$ and $J \subseteq [6]_0$,
$$\sum_{i \in \bigcup_{t \in J}S_t} \left( w_i \sum_{t \in J}p_t \right) \le \sum_{i \in \bigcup_{t \in J}S_t}w_i \le 1.$$

However, the left hand side can be written as:
$$\sum_{i \in \bigcup_{t \in J}S_t} \left( w_i \sum_{t \in J}p_t \right)= \sum_{t=0}^5  \left(p_t \sum_{i \in S_t}w_i \right).$$

If the weight profile is supporting, i.e., $ \sum_{i \in S_t}w_i > \frac{1}{2}$ for each $t$, then $$ \sum_{t=0}^5  \left(p_t \sum_{i \in S_t}w_i \right) > \frac{1}{2} \sum_{t=0}^5 p_t = 1,$$
a contradiction. Hence the claim follows.

Since the cases for $t=3$ or $t=4$ are symmetric to cases for $t=0$ or $t=1$, we consider two cases: $t=0$ or $t=1$.

If an agent for $t=0$ exists, then we have $\alpha < 2.0789$ by \Cref{lem:0134-025}.

Assume an agent for $t=1$ exists, i.e. $i' \in S_1 \cap S_2 \cap S_5$. Note that there must exist an agent that is included in three of $S_2 \cap S_3 \cap S_4 \cap S_5$. We perform case analysis.

If $i'' \in S_2 \cap S_3 \cap S_4$, then by taking $i_0=i_3=i$, $i_1=i'$, $i_2=i_3=i''$ and applying \Cref{lem:deviation-step} and \Cref{lem:cases-full} yields

$$d(i_0, c_0) < \left(\frac{1}{\alpha} \right) \left(\frac{\alpha^2+1}{\alpha^2 ( \alpha - 1)^2}  \right)^2  \left(\frac{\alpha+1}{\alpha ( \alpha - 1)}  \right)  d(i_0, c_0),$$
which implies $\alpha < \overline{\alpha} \approx 2.0264$.

If $i'' \in S_3 \cap S_4 \cap S_5$, then $\alpha < 2$ by \Cref{lem:0134-125-345}.

If $i'' \in S_2 \cap S_3 \cap S_5$, then $\alpha< 2.0567$ by \Cref{lem:0134-125-235}.

Finally if $i'' \in S_2 \cap S_4 \cap S_5$, then $\alpha < \alpha_6$ by \Cref{lem:0134-125-245}.

To show tightness, let $I = (\{i_0,i_1,i_2,i_3,i_4\}, \{c_0,c_1,c_2,c_3,c_4,c_5\}, d, w)$ where $w_{i_0}=w_{i_3}=2/9$, $w_{i_1}=w_{i_2}=1/9$, $w_{i_4}=1/3$ (equivalently, agents $i_0$ and $i_3$ have $2$ copies each, and agent $i_4$ has $3$ copies), and $d$ is given in \Cref{fig:T6-instance}.

\begin{figure}[h]
    \centering
    \begin{subfigure}[c]{0.44\linewidth}
        \centering
        \scalebox{0.75}{$\begin{array}{c|cccccc}
            & c_0 & c_1 & c_2 & c_3 & c_4 & c_5\\
            \hline
            i_0 &
            \gamma\delta &
            \alpha_6^3\gamma^2 &
            \alpha_6^2\gamma^2 &
            \alpha_6\gamma^2 &
            \alpha_6^2\gamma &
            \alpha_6\gamma\delta
            \\[0.8em]
            i_1 &
            \alpha_6^2 &
            \alpha_6 &
            \alpha_6^3 &
            \alpha_6\xi &
            \xi &
            \alpha_6^3
            \\[0.8em]
            i_2 &
            \alpha_6\xi &
            \xi &
            \alpha_6^3 &
            \alpha_6^2 &
            \alpha_6 &
            \alpha_6^3-2\alpha_6\gamma
            \\[0.8em]
            i_3 &
            \alpha_6\gamma^2 &
            \alpha_6^2\gamma &
            \alpha_6\gamma\delta &
            \gamma\delta &
            \alpha_6^3\gamma^2 &
            \alpha_6^2\gamma^2
            \\[0.8em]
            i_4 &
            \alpha_6^3\gamma &
            \alpha_6^2\gamma &
            \alpha_6\gamma &
            \alpha_6^3\gamma &
            \alpha_6^2\gamma &
            \alpha_6\gamma
        \end{array}$}
        \caption{Distance matrix $[d(i, c)]$ between each agent and center.}
    \end{subfigure}
    \hfill
    \begin{subfigure}[c]{0.33\linewidth}
        \centering
        \begin{tikzpicture}[
          cnode/.style={draw, circle, thick, minimum size=7mm, inner sep=1pt, fill=white},
          >=Stealth, font=\small,
        ]
        \node[cnode] (c0) at (1.2,  2.4)  {$c_0$};
        \node[cnode] (c1) at (2.24, 1.8)  {$c_1$};
        \node[cnode] (c2) at (2.24, 0.6)  {$c_2$};
        \node[cnode] (c3) at (1.2,  0.0)  {$c_3$};
        \node[cnode] (c4) at (0.16, 0.6)  {$c_4$};
        \node[cnode] (c5) at (0.16, 1.8)  {$c_5$};
        \draw[->, thick] (c0) -- node[above right, inner sep=1pt] {$S_1$} (c1);
        \draw[->, thick] (c1) -- node[right,        inner sep=1pt] {$S_2$} (c2);
        \draw[->, thick] (c2) -- node[below right, inner sep=1pt] {$S_3$} (c3);
        \draw[->, thick] (c3) -- node[below left,  inner sep=1pt] {$S_4$} (c4);
        \draw[->, thick] (c4) -- node[left,         inner sep=1pt] {$S_5$} (c5);
        \draw[->, thick] (c5) -- node[above left,  inner sep=1pt] {$S_0$} (c0);
        \end{tikzpicture}
        \caption{Deviation graph $G_\alpha$: the directed 6-cycle $c_0  \cdots c_5c_0$, with deviating set $S_t$ labelling each edge.}
    \end{subfigure}
    \hfill
    \begin{subfigure}[c]{0.33\linewidth}
        \centering
        $\begin{array}{cc|cccccc}
            w_i & & S_0 & S_1 & S_2 & S_3 & S_4 & S_5 \\
            \hline
            \nicefrac{2}{9} & i_0 & 1 & 0 & 1 & 1 & 0 & 0 \\
            \nicefrac{1}{9} & i_1 & 1 & 1 & 0 & 0 & 1 & 0 \\
            \nicefrac{1}{9} & i_2 & 0 & 1 & 0 & 1 & 1 & 0 \\
            \nicefrac{2}{9} & i_3 & 1 & 0 & 0 & 1 & 0 & 1 \\
            \nicefrac{1}{3} & i_4 & 0 & 1 & 1 & 0 & 1 & 1 \\
            \hline
            & \text{sum} & \nicefrac{5}{9} & \nicefrac{5}{9} & \nicefrac{5}{9} & \nicefrac{5}{9} & \nicefrac{5}{9} & \nicefrac{5}{9}
        \end{array}$
        \caption{Set membership: entry is $1$ iff agent $i_t$ belongs to deviating set $S_t$ for all $\alpha<\alpha_6$; bottom row gives the total weight of each set.}
    \end{subfigure}
    \caption{Tightness instance for $m = 6$: the distance matrix (a), the resulting deviation graph (b), and the deviating set memberships (c). Here $\gamma = \alpha_6 - 1$, $\delta = \alpha_6^2-\alpha_6+2$, and $\xi = \alpha_6^3(\alpha_6-1)^2-\alpha_6^2(\alpha_6-1)-\alpha_6$.}
    \label{fig:T6-instance}
\end{figure}
Then, for every $c_t$ and every $\alpha < \alpha_6$, a strict majority of agents $\alpha$-deviate to $c_{t+1}$, so $G_\alpha(I)$ is the directed cycle $c_0c_1c_2c_3c_4c_5c_0$.
\end{proof}

Now we state the lemmas used in the proof.

\begin{lemma}\label{lem:0124}
	Let $G_\alpha$ contain a directed cycle $c_0c_1\ldots c_5c_0$ of length 6.
	If there exists $i \in S_0 \cap S_1 \cap S_2 \cap S_4$, then $\alpha < 2.012$.
\end{lemma}

\begin{proof}
	Take $i_0 = i_1 = i$.
	Let \(\overline{\alpha}\) be the unique real root in \((2, \infty)\) of
	\[
	p(z)=z^8-4z^7+7z^6-8z^5+6z^4-4z^3+z^2-2z-1.
	\]
	Numerically, $\overline{\alpha}\approx 2.01174$.
	Suppose, for contradiction, that the $G_\alpha$ contains a directed cycle of length 6 for some
	\(\alpha\ge \overline{\alpha}\). Then $G_{\overline{\alpha}}$ must also contain the cycle.

	Define the following \(3\)-path residuals:
	\begin{align*}
		T_1&=d(i_2,c_2)+d(i,c_2)+d(i,c_1)-d(i_2,c_1),\\
		T_2&=d(i_3,c_4)+d(i,c_4)+d(i,c_2)-d(i_3,c_2),\\
		T_3&=d(i_4,c_4)+d(i_3,c_4)+d(i_3,c_3)-d(i_4,c_3),\\
		T_4&=d(i,c_2)+d(i_2,c_2)+d(i_2,c_3)-d(i,c_3),\\
		T_5&=d(i,c_4)+d(i_4,c_4)+d(i_4,c_5)-d(i,c_5).
	\end{align*}
	Then each \(T_j\ge 0\) by the \(3\)-path inequalities.

	Now define
	\[
	K=\overline{\alpha}^3-2\overline{\alpha}^2+2\overline{\alpha}+1
	\]
	and
	\[
	L=\overline{\alpha}^5(\overline{\alpha}-1)^3-(\overline{\alpha}+1)K.
	\]
	Define the deviation residuals:
	\begin{align*}
		\Delta_1&=d(i_2,c_1)-\overline{\alpha} d(i_2,c_2),\\
		\Delta_2&=d(i_2,c_2)-\overline{\alpha} d(i_2,c_3),\\
		\Delta_3&=d(i_3,c_2)-\overline{\alpha} d(i_3,c_3),\\
		\Delta_4&=d(i_3,c_3)-\overline{\alpha} d(i_3,c_4),\\
		\Delta_5&=d(i_4,c_3)-\overline{\alpha} d(i_4,c_4),\\
		\Delta_6&=d(i_4,c_4)-\overline{\alpha} d(i_4,c_5),\\
		\Delta_7&=d(i,c_5)-\overline{\alpha} d(i,c_0),\\
		\Delta_8&=d(i,c_0)-\overline{\alpha} d(i,c_1),\\
		\Delta_9&=d(i,c_1)-\overline{\alpha} d(i,c_2),\\
		\Delta_{10}&=d(i,c_3)-\overline{\alpha} d(i,c_4).
	\end{align*}
	Then each \(\Delta_j>0\) by the strict \(\overline{\alpha}\)-deviation inequalities.

	Now consider the weighted sum
	\begin{align*}
		\Phi={}&
		(\overline{\alpha}+1)K(T_1+\Delta_1)\\
		&+\overline{\alpha}^2(\overline{\alpha}-1)(\overline{\alpha}+1)(T_2+\Delta_3+\Delta_4)\\
		&+\overline{\alpha}^2(\overline{\alpha}-1)^2(\overline{\alpha}+1)(T_3+\Delta_5)\\
		&+\overline{\alpha}(\overline{\alpha}-1)K(T_4+\Delta_{10})\\
		&+\overline{\alpha}^3(\overline{\alpha}-1)^3(T_5+\Delta_7)\\
		&+(\overline{\alpha}-1)K\Delta_2\\
		&+\overline{\alpha}^2(\overline{\alpha}-1)^3\Delta_6\\
		&+\overline{\alpha}^4(\overline{\alpha}-1)^3\Delta_8\\
		&+L\Delta_9.
	\end{align*}

	All coefficients appearing in \(\Phi\) are positive. Indeed
	\(\overline{\alpha}>2\), so \(K, L>0\).

	Therefore, since each \(T_j\ge 0\) and each \(\Delta_j>0\), we have
	\[
	\Phi>0.
	\]

	On the other hand, expanding \(\Phi\) as a linear combination of the variables
	\(d(i_r,c_s)\) gives
	\[
	\Phi=-(\overline{\alpha}+1)p(\overline{\alpha})d(i,c_2).
	\]
	Since \(p(\overline{\alpha})=0\), this implies
	\[
	\Phi=0,
	\]
	contradicting \(\Phi>0\). Hence $G_\alpha$ does not contain a 6-cycle for every
	\(\alpha\ge \overline{\alpha}\).
\end{proof}

\begin{lemma}\label{lem:0134-1245}
	Let $G_\alpha$ contain a directed cycle $c_0c_1\ldots c_5c_0$ of length 6.
	If there exist $i \in S_0 \cap S_1 \cap S_3 \cap S_4$ and $i' \in S_1 \cap S_2 \cap S_4 \cap S_5$, then $\alpha < 2$.
\end{lemma}
\begin{proof}
	Take $i_0 = i_3 = i$, $i_1 = i_4 = i'$ and $i_2 \in S_2 \cap S_3$, $i_5 \in S_5 \cap S_0$.
	Define the following $3$-path residuals:
	\begin{align*}
		T_1&=d(i,c_4)+d(i',c_4)+d(i',c_2)-d(i,c_2),\\
		T_2&=d(i,c_1)+d(i',c_1)+d(i',c_5)-d(i,c_5),\\
		T_3&=d(i',c_5)+d(i_5,c_5)+d(i_5,c_0)-d(i',c_0),\\
		T_4&=d(i',c_2)+d(i_2,c_2)+d(i_2,c_3)-d(i',c_3),\\
		T_5&=d(i_2,c_3)+d(i,c_3)+d(i,c_1)-d(i_2,c_1),\\
		T_6&=d(i_5,c_0)+d(i,c_0)+d(i,c_4)-d(i_5,c_4).
	\end{align*}
	Then each $T_j\ge 0$ by the $3$-path inequalities.

	Now define the deviation residuals:
	\begin{align*}
		\Delta_1&=d(i,c_5)-\alpha d(i,c_0),\\
		\Delta_2&=d(i,c_0)-\alpha d(i,c_1),\\
		\Delta_3&=d(i,c_2)-\alpha d(i,c_3),\\
		\Delta_4&=d(i,c_3)-\alpha d(i,c_4),\\
		\Delta_5&=d(i',c_0)-\alpha d(i',c_1),\\
		\Delta_6&=d(i',c_1)-\alpha d(i',c_2),\\
		\Delta_7&=d(i',c_3)-\alpha d(i',c_4),\\
		\Delta_8&=d(i',c_4)-\alpha d(i',c_5),\\
		\Delta_9&=d(i_2,c_1)-\alpha d(i_2,c_2),\\
		\Delta_{10}&=d(i_2,c_2)-\alpha d(i_2,c_3),\\
		\Delta_{11}&=d(i_5,c_4)-\alpha d(i_5,c_5),\\
		\Delta_{12}&=d(i_5,c_5)-\alpha d(i_5,c_0).
	\end{align*}
	Then each $\Delta_j>0$ by the strict deviation inequalities.

	Now consider the sum
	\begin{align*}
		\Phi
		={}&
		T_1+T_2+T_3+T_4+T_5+T_6\\
		&+\Delta_1+\Delta_2+\Delta_3+\Delta_4+\Delta_5+\Delta_6\\
		&+\Delta_7+\Delta_8+\Delta_9+\Delta_{10}+\Delta_{11}+\Delta_{12}.
	\end{align*}

	Since each $T_j\ge 0$ and each $\Delta_j>0$, we have
	$$
	\Phi>0.
	$$

	On the other hand, expanding $\Phi$ as a linear combination of the variables $d(i_r,c_s)$ gives
	\begin{align*}
		\Phi
		={}&(2-\alpha)\Bigl(
		d(i,c_0)+d(i,c_1)+d(i,c_3)+d(i,c_4)\\
		&\qquad\qquad
		+d(i',c_1)+d(i',c_2)+d(i',c_4)+d(i',c_5)\\
		&\qquad\qquad
		+d(i_2,c_2)+d(i_2,c_3)+d(i_5,c_0)+d(i_5,c_5)
		\Bigr).
	\end{align*}

	Since distances are nonnegative, this implies $\alpha < 2$.
\end{proof}

\begin{lemma}\label{lem:0134-025}
	Let $G_\alpha$ contain a directed cycle $c_0c_1\ldots c_5c_0$ of length 6.
	If there exist $i \in S_0 \cap S_1 \cap S_3 \cap S_4$ and $i' \in S_0 \cap S_2 \cap S_5$, then $\alpha < 2.0789$.
\end{lemma}
\begin{proof}
	Take $i_0 = i_3 = i$, $i_5 = i'$.

	Let \(\overline{\alpha}\) be the unique real root \(>2\) of
	\[
	p(z)=z^8-4z^7+7z^6-8z^5+6z^4-6z^3+3z^2-2z+1.
	\]
	Numerically, $\overline{\alpha}\approx 2.07887$.

	Suppose, for contradiction, that the $G_\alpha$ contains a directed cycle of length 6 for some
	\(\alpha\ge \overline{\alpha}\). Then $G_{\overline{\alpha}}$ must also contain the cycle.

	Define the following \(3\)-path residuals:
	\begin{align*}
		T_1&=d(i_1,c_2)+d(i',c_2)+d(i',c_0)-d(i_1,c_0),\\
		T_2&=d(i_4,c_4)+d(i,c_4)+d(i,c_3)-d(i_4,c_3),\\
		T_3&=d(i',c_0)+d(i,c_0)+d(i,c_1)-d(i',c_1),\\
		T_4&=d(i',c_0)+d(i,c_0)+d(i,c_4)-d(i',c_4),\\
		T_5&=d(i,c_1)+d(i_1,c_1)+d(i_1,c_2)-d(i,c_2),\\
		T_6&=d(i,c_4)+d(i_4,c_4)+d(i_4,c_5)-d(i,c_5).
	\end{align*}
	Then each \(T_j\ge 0\) by the \(3\)-path inequalities.

	Now define the deviation residuals:
	\begin{align*}
		\Delta_1&=d(i_1,c_0)-\overline{\alpha}d(i_1,c_1),\\
		\Delta_2&=d(i_1,c_1)-\overline{\alpha}d(i_1,c_2),\\
		\Delta_3&=d(i_4,c_3)-\overline{\alpha}d(i_4,c_4),\\
		\Delta_4&=d(i',c_1)-\overline{\alpha}d(i',c_2),\\
		\Delta_5&=d(i',c_4)-\overline{\alpha}d(i',c_5),\\
		\Delta_6&=d(i,c_2)-\overline{\alpha}d(i,c_3),\\
		\Delta_7&=d(i,c_5)-\overline{\alpha}d(i,c_0),\\
		\Delta_8&=d(i_4,c_4)-\overline{\alpha}d(i_4,c_5),\\
		\Delta_9&=d(i',c_5)-\overline{\alpha}d(i',c_0),\\
		\Delta_{10}&=d(i,c_0)-\overline{\alpha}d(i,c_1),\\
		\Delta_{11}&=d(i,c_3)-\overline{\alpha}d(i,c_4).
	\end{align*}
	Then each \(\Delta_j>0\) by the strict
	\(\overline{\alpha}\)-deviation inequalities.

	Now define
	\[
	P=\overline{\alpha}^3-\overline{\alpha}^2+2\overline{\alpha}-1
	\]
	and
	\[
	Q=\overline{\alpha}^5-3\overline{\alpha}^4
	+4\overline{\alpha}^3-3\overline{\alpha}^2
	+3\overline{\alpha}-1.
	\]
	Then, $P,Q > 0$ as \(\overline{\alpha}>2\).

	Consider the weighted sum
	\begin{align*}
		\Phi={}&
		\overline{\alpha}^4(\overline{\alpha}-1)^2
		\left(T_1+\Delta_1+\Delta_2\right)\\
		&+(\overline{\alpha}+1)P
		\left(T_2+\Delta_3\right)\\
		&+\overline{\alpha}^3(\overline{\alpha}-1)^2
		\left(T_3+\Delta_4\right)\\
		&+\overline{\alpha}^3(\overline{\alpha}-1)
		\left(T_4+\Delta_5\right)\\
		&+\overline{\alpha}^4(\overline{\alpha}-1)^3
		\left(T_5+\Delta_6\right)\\
		&+\overline{\alpha}(\overline{\alpha}-1)P
		\left(T_6+\Delta_7\right)\\
		&+(\overline{\alpha}-1)P\Delta_8\\
		&+\overline{\alpha}^4(\overline{\alpha}-1)\Delta_9\\
		&+\overline{\alpha}^2(\overline{\alpha}-1)^2
		(\overline{\alpha}^2-\overline{\alpha}+1)\Delta_{10}\\
		&+\overline{\alpha}(\overline{\alpha}-1)Q\Delta_{11}.
	\end{align*}

	All coefficients appearing in \(\Phi\) are positive.
	Therefore,
	\[
	\Phi>0.
	\]

	On the other hand, expanding \(\Phi\) as a linear combination of the variables
	\(d(i_r,c_s)\) gives
	\[
	\Phi
	=
	-p(\overline{\alpha})
	\left(
	d(i,c_3)+d(i,c_4)
	\right).
	\]
	Since \(p(\overline{\alpha})=0\), this implies
	\[
	\Phi=0,
	\]
	contradicting \(\Phi>0\).  Hence $G_\alpha$ does not contain a 6-cycle for every
	\(\alpha\ge \overline{\alpha}\).
\end{proof}

\begin{lemma}\label{lem:0134-125-345}
	Let $G_\alpha$ contain a directed cycle $c_0c_1\ldots c_5c_0$ of length 6.
	If there exist $i \in S_0 \cap S_1 \cap S_3 \cap S_4$, $i' \in S_1 \cap S_2 \cap S_5$, and $i'' \in S_3 \cap S_4 \cap S_5$, then $\alpha < 2$.
\end{lemma}
\begin{proof}
	Take $i_0 = i$, $i_1 = i'$, and $i_3 = i_4 = i''$.

	Suppose, for contradiction, that the $G_\alpha$ contains a directed cycle of length 6 for some \(\alpha\ge 2\). Then $G_{2}$ must also contain the cycle.

	Define the following \(3\)-path residuals:
	\begin{align*}
		T_1&=d(i',c_5)+d(i_5,c_5)+d(i_5,c_0)-d(i',c_0),\\
		T_2&=d(i',c_5)+d(i'',c_5)+d(i'',c_4)-d(i',c_4),\\
		T_3&=d(i'',c_5),\\
		T_4&=d(i'',c_5)+d(i',c_5)+d(i',c_2)-d(i'',c_2),\\
		T_5&=d(i_5,c_5)+d(i'',c_5)+d(i'',c_4)-d(i_5,c_4).
	\end{align*}
	Then each \(T_r\ge 0\) by the \(3\)-path inequalities or distances being nonnegative.

	Now define the \(2\)-deviation residuals:
	\begin{align*}
		\Delta_1&=d(i',c_0)-2d(i',c_1),\\
		\Delta_2&=d(i',c_1)-2d(i',c_2),\\
		\Delta_3&=d(i',c_4)-2d(i',c_5),\\
		\Delta_4&=d(i'',c_2)-2d(i'',c_3),\\
		\Delta_5&=d(i'',c_4)-2d(i'',c_5),\\
		\Delta_6&=d(i'',c_3)-2d(i'',c_4),\\
		\Delta_7&=d(i_5,c_5)-2d(i_5,c_0),\\
		\Delta_8&=d(i_5,c_4)-2d(i_5,c_5).
	\end{align*}
	Then each \(\Delta_s>0\) by the strict deviation inequalities.

	Consider the weighted sum
	\begin{align*}
		\Phi={}&
		4T_1+20T_2+34T_3+16T_4+6T_5\\
		&+4\Delta_1+8\Delta_2+20\Delta_3+16\Delta_4\\
		&+38\Delta_5+32\Delta_6+2\Delta_7+6\Delta_8.
	\end{align*}
	All coefficients are positive, each \(T_r\ge 0\), and each \(\Delta_s>0\).
	Hence
	\[
	\Phi>0.
	\]

	On the other hand, expanding \(\Phi\) gives exact cancellation of every
	variable:
	\[
	\Phi=0.
	\]
	This is a contradiction. Therefore $\alpha < 2$.
\end{proof}

\begin{lemma}\label{lem:0134-125-235}
	Let $G_\alpha$ contain a directed cycle $c_0c_1\ldots c_5c_0$ of length 6.
	If there exist $i \in S_0 \cap S_1 \cap S_3 \cap S_4$, $i' \in S_1 \cap S_2 \cap S_5$, and $i'' \in S_2 \cap S_3 \cap S_5$, then $\alpha < 2.0567$.
\end{lemma}
\begin{proof}
	Take $i_0 = i_3 = i$, $i_1 = i'$, and $i_2 = i''$.

	Let \(\overline{\alpha}\) be the unique real root \(>2\) of
	\[
	q(z)=z^9-3z^8+3z^7-z^6-3z^5+3z^4-4z^3+2z^2-4z+2.
	\]
	Numerically, $\overline{\alpha}\approx 2.05664$.

	Suppose, for contradiction, that the $G_\alpha$ contains a directed cycle of length 6 for some \(\alpha\ge \overline{\alpha}\). Then $G_{\overline{\alpha}}$ must also contain the cycle.

	Define the following \(3\)-path residuals:
	\begin{align*}
		T_1&=d(i,c_4)+d(i_4,c_4)+d(i_4,c_5)-d(i,c_5),\\
		T_2&=d(i',c_5)+d(i_5,c_5)+d(i_5,c_0)-d(i',c_0),\\
		T_3&=d(i',c_1)+d(i,c_1)+d(i,c_4)-d(i',c_4),\\
		T_4&=d(i'',c_3)+d(i,c_3)+d(i,c_1)-d(i'',c_1),\\
		T_5&=d(i'',c_3)+d(i,c_3)+d(i,c_4)-d(i'',c_4),\\
		T_6&=d(i,c_1)+d(i',c_1)+d(i',c_2)-d(i,c_2),\\
		T_7&=d(i_4,c_5)+d(i'',c_5)+d(i'',c_3)-d(i_4,c_3),\\
		T_8&=d(i_5,c_0)+d(i,c_0)+d(i,c_4)-d(i_5,c_4).
	\end{align*}
	Then each \(T_j\ge 0\) by the \(3\)-path inequalities.

	Now define the deviation residuals:
	\begin{align*}
		\Delta_1&=d(i,c_5)-\overline{\alpha}d(i,c_0),\\
		\Delta_2&=d(i,c_0)-\overline{\alpha}d(i,c_1),\\
		\Delta_3&=d(i',c_0)-\overline{\alpha}d(i',c_1),\\
		\Delta_4&=d(i',c_1)-\overline{\alpha}d(i',c_2),\\
		\Delta_5&=d(i',c_4)-\overline{\alpha}d(i',c_5),\\
		\Delta_6&=d(i'',c_1)-\overline{\alpha}d(i'',c_2),\\
		\Delta_7&=d(i'',c_2)-\overline{\alpha}d(i'',c_3),\\
		\Delta_8&=d(i'',c_4)-\overline{\alpha}d(i'',c_5),\\
		\Delta_9&=d(i,c_2)-\overline{\alpha}d(i,c_3),\\
		\Delta_{10}&=d(i,c_3)-\overline{\alpha}d(i,c_4),\\
		\Delta_{11}&=d(i_4,c_3)-\overline{\alpha}d(i_4,c_4),\\
		\Delta_{12}&=d(i_4,c_4)-\overline{\alpha}d(i_4,c_5),\\
		\Delta_{13}&=d(i_5,c_5)-\overline{\alpha}d(i_5,c_0),\\
		\Delta_{14}&=d(i_5,c_4)-\overline{\alpha}d(i_5,c_5).
	\end{align*}
	Then each \(\Delta_j>0\) by the strict
	\(\overline{\alpha}\)-deviation inequalities.

	Now define
	\begin{align*}
		P&=\overline{\alpha}^4+\overline{\alpha}^3-\overline{\alpha}^2+2\overline{\alpha}-1, \\
		R&=\overline{\alpha}^6-3\overline{\alpha}^5
		+4\overline{\alpha}^4-4\overline{\alpha}^3
		+3\overline{\alpha}^2-3\overline{\alpha}+1,\\
		S&=2\overline{\alpha}^5-2\overline{\alpha}^4
		+2\overline{\alpha}^3+\overline{\alpha}^2
		+3\overline{\alpha}-2,\\
		Z&=\overline{\alpha}^4-\overline{\alpha}^3+\overline{\alpha}^2
		+4\overline{\alpha}-2,\\
		W&=2\overline{\alpha}^6-4\overline{\alpha}^5
		+3\overline{\alpha}^4-2\overline{\alpha}^3
		+3\overline{\alpha}^2-7\overline{\alpha}+3.
	\end{align*}

	Then consider the weighted sum
	\begin{align*}
		\Phi={}&
		\overline{\alpha}^3(\overline{\alpha}-1)^2P(T_1+\Delta_1)\\
		&+\overline{\alpha}^3(\overline{\alpha}+1)(\overline{\alpha}-1)RT_2\\
		&+\overline{\alpha}^2(\overline{\alpha}+1)(\overline{\alpha}-1)R(T_3+\Delta_5)\\
		&+\overline{\alpha}^2P(T_4+\Delta_6)\\
		&+\overline{\alpha}^2(\overline{\alpha}-1)P(T_5+\Delta_8)\\
		&+\overline{\alpha}^2(\overline{\alpha}-1)S(T_6+\Delta_9)\\
		&+\overline{\alpha}^3(\overline{\alpha}-1)P(T_7+\Delta_{11}+\Delta_{12})\\
		&+\overline{\alpha}^3(\overline{\alpha}+1)R(T_8+\Delta_{13}+\Delta_{14})\\
		&+\overline{\alpha}^3(R+Z)\Delta_2\\
		&+(\overline{\alpha}-1)(\overline{\alpha}+1)(\overline{\alpha}R+S)\Delta_3\\
		&+\overline{\alpha}(\overline{\alpha}-1)S\Delta_4\\
		&+\overline{\alpha}^3P\Delta_7\\
		&+\overline{\alpha}^3W\Delta_{10}.
	\end{align*}

	We claim that all coefficients appearing in \(\Phi\) are positive. Clearly
	\(\overline{\alpha}>2\). Then we have \(P,R,S,Z,W>0\), and hence every coefficient in \(\Phi\) is positive.
	Since each \(T_j\ge 0\) and each \(\Delta_j>0\), it follows that
	\[
	\Phi>0.
	\]

	On the other hand, expanding \(\Phi\) as a linear combination of the variables
	\(d(\cdot,\cdot)\) gives
	\[
	\Phi
	=
	-(\overline{\alpha}-1)(\overline{\alpha}+1)
	q(\overline{\alpha})d(i',c_0).
	\]
	Since \(q(\overline{\alpha})=0\), this implies
	\[
	\Phi=0,
	\]
	contradicting \(\Phi>0\).
	Hence $G_\alpha$ does not contain a 6-cycle for every
	\(\alpha\ge \overline{\alpha}\).
\end{proof}

\begin{lemma}\label{lem:0134-125-245}
	Let $G_\alpha$ contain a directed cycle $c_0c_1\ldots c_5c_0$ of length 6.
	If there exist $i \in S_0 \cap S_1 \cap S_3 \cap S_4$, $i' \in S_1 \cap S_2 \cap S_5$, and $i'' \in S_2 \cap S_4 \cap S_5$, then $\alpha < \alpha_6$.
\end{lemma}
\begin{proof}
	Take $i_0=i_3 = i$, $i_1 = i'$, and $i_4 = i''$.

	Suppose, for contradiction, that the $G_\alpha$ contains a directed cycle of length 6 for some \(\alpha\ge \alpha_6\). Then $G_{\alpha_6}$ must also contain the cycle.

	Define the following \(3\)-path residuals:
	\begin{align*}
		T_1&=d(i',c_5)+d(i_5,c_5)+d(i_5,c_0)-d(i',c_0),\\
		T_2&=d(i',c_1)+d(i,c_1)+d(i,c_4)-d(i',c_4),\\
		T_3&=d(i'',c_3)+d(i,c_3)+d(i,c_1)-d(i'',c_1),\\
		T_4&=d(i,c_1)+d(i',c_1)+d(i',c_2)-d(i,c_2),\\
		T_5&=d(i,c_4)+d(i_4,c_4)+d(i_4,c_5)-d(i,c_5),\\
		T_6&=d(i_4,c_4)+d(i,c_4)+d(i,c_1)-d(i_4,c_1),\\
		T_7&=d(i_4,c_2)+d(i'',c_2)+d(i'',c_3)-d(i_4,c_3),\\
		T_8&=d(i_5,c_0)+d(i,c_0)+d(i,c_4)-d(i_5,c_4).
	\end{align*}
	Then each \(T_r\ge 0\) by the \(3\)-path inequalities.

	Now define the deviation residuals:
	\begin{align*}
		\Delta_1&=d(i,c_2)-{\alpha_6}d(i,c_3),\\
		\Delta_2&=d(i,c_3)-{\alpha_6}d(i,c_4),\\
		\Delta_3&=d(i',c_0)-{\alpha_6}d(i',c_1),\\
		\Delta_4&=d(i',c_1)-{\alpha_6}d(i',c_2),\\
		\Delta_5&=d(i',c_4)-{\alpha_6}d(i',c_5),\\
		\Delta_6&=d(i'',c_1)-{\alpha_6}d(i'',c_2),\\
		\Delta_7&=d(i'',c_2)-{\alpha_6}d(i'',c_3),\\
		\Delta_8&=d(i,c_5)-{\alpha_6}d(i,c_0),\\
		\Delta_9&=d(i,c_0)-{\alpha_6}d(i,c_1),\\
		\Delta_{10}&=d(i_4,c_1)-{\alpha_6}d(i_4,c_2),\\
		\Delta_{11}&=d(i_4,c_3)-{\alpha_6}d(i_4,c_4),\\
		\Delta_{12}&=d(i_4,c_4)-{\alpha_6}d(i_4,c_5),\\
		\Delta_{13}&=d(i_5,c_5)-{\alpha_6}d(i_5,c_0),\\
		\Delta_{14}&=d(i_5,c_4)-{\alpha_6}d(i_5,c_5).
	\end{align*}
	Then each \(\Delta_s>0\) by the strict
	\({\alpha_6}\)-deviation inequalities.

	Define
	\[
	M={\alpha_6}^4-{\alpha_6}^3-{\alpha_6}^2-2
	\]
	and
	\[
	N=(2{\alpha_6}+2-{\alpha_6}^2)
	({\alpha_6}^2-{\alpha_6}+1).
	\]

	Consider the weighted sum
	\begin{align*}
		\Phi={}&
		2{\alpha_6}T_1
		+2T_2
		+MT_3
		+2{\alpha_6}({\alpha_6}-1)T_4
		+2{\alpha_6}({\alpha_6}-1)T_5\\
		&+2T_6
		+2{\alpha_6}T_7
		+MT_8\\
		&+2{\alpha_6}({\alpha_6}-1)\Delta_1
		+N\Delta_2
		+2{\alpha_6}\Delta_3
		+2({\alpha_6}-1)\Delta_4
		+2\Delta_5\\
		&+M\Delta_6
		+M\Delta_7
		+2{\alpha_6}({\alpha_6}-1)\Delta_8
		+N\Delta_9
		+2\Delta_{10}\\
		&+2{\alpha_6}\Delta_{11}
		+2({\alpha_6}-1)\Delta_{12}
		+M\Delta_{13}
		+M\Delta_{14}.
	\end{align*}

	We claim that all coefficients in \(\Phi\) are positive. Clearly
	\({\alpha_6}>2\), which implies $M > 0$. Also \(2<{\alpha_6}<2.11<1+\sqrt{3}\), so
	\[
	2{\alpha_6}+2-{\alpha_6}^2>0,
	\]
	and clearly
	\[
	{\alpha_6}^2-{\alpha_6}+1>0.
	\]
	Thus \(N>0\). Therefore every coefficient in \(\Phi\) is positive.

	Since each \(T_r\ge 0\) and each \(\Delta_s>0\), it follows that
	\[
	\Phi>0.
	\]

	On the other hand, expanding \(\Phi\) as a linear combination of the variables
	\(d(\cdot,\cdot)\) gives
	\[
	\Phi
	=
	P_6({\alpha_6})
	\left(
	d(i,c_1)+d(i,c_4)
	-d(i'',c_2)-d(i'',c_3)
	-d(i_5,c_0)-d(i_5,c_5)
	\right).
	\]
	Since \(P_6({\alpha_6})=0\), this implies
	\[
	\Phi=0,
	\]
	contradicting \(\Phi>0\). Hence $G_\alpha$ does not contain a 6-cycle for every
	\(\alpha\ge \alpha_6 \).
\end{proof}

\section{Computational certificate for the 37-point instance}
\label{app:metric37-certificate}

We report the numerical certificate obtained for the 37-point
instance used in our construction. The accompanying repository\footnote{\url{https://github.com/evamichelle30/ProportionalClusteringLowerBounds}} contains the
code and input files needed to generate the instance and reproduce the
feasibility computation.

The instance consists of a finite (pseudo-)metric space $(N\cup C,d)$ with $C=\{c_0,c_1,\ldots,c_{36}\}$. We define a family of deviating sets $\mathcal{S}=\{S_0,S_1,\ldots,S_{36}\}$.
The computation searches for the largest feasible value of the parameter
$\alpha$. A binary search over $[1,2.42]$ returned the feasible value $\alpha = 2.15082963$.
At this value, the solver produced a nonnegative weight vector $w=(w_0,w_1,\ldots,w_{36})\in\mathbb{R}_{\ge 0}^{37}$
with $\sum_{i=0}^{36} w_i = 1$.
For every set $S_j$, the total weight assigned to agents contained in $S_j$
is at least $0.50000100>\frac{1}{2}$. Therefore, the vector $w$ is a feasible certificate for the reported value of $\alpha$.

\begin{table}[t]
	\centering
	\begin{tabular}{lr}
		\toprule
		Quantity & Value \\
		\midrule
		Number of center points & $37$ \\
		Number of single agents & $1000000$ \\
		Search interval for $\alpha$ & $[1,2.42]$ \\
		Feasible value of $\alpha$ & $2.15082963$ \\
		Minimum set coverage & $0.50000100$ \\
		\bottomrule
	\end{tabular}
	\caption{Summary of the computational certificate for the 37-point instance.}
	\label{tab:metric37-summary}
\end{table}

\Cref{tab:metric37-agents-by-weight} reports the feasible weight vector in
decreasing order of weight, together with the sets containing each agent. The
two largest weights are assigned to agents $35$ and $36$, which together carry
approximately $81.4\%$ of the total mass.

	\begin{table}[p]
		\centering
		\scriptsize
		\begin{tabular}{r r l}
			\toprule
			Agent $i$ & Weight $w_i$ & Sets containing agent $i$ \\
			\midrule
			35 & $0.47712300$ & $\{0,2,4,6,8,10,12,14,16,18,20,22,24,26,28,30,32,34,36\}$ \\
			36 & $0.33711700$ & $\{0,1,3,5,7,9,11,13,15,17,19,21,23,25,27,29,31,33,35\}$ \\
			32 & $0.01354200$ & $\{1,3,5,7,9,11,13,15,17,19,21,23,25,27,29,31,33,34,36\}$ \\
			30 & $0.01228800$ & $\{1,3,5,7,9,11,13,15,17,19,21,23,25,27,29,31,32,35\}$ \\
			28 & $0.01142800$ & $\{1,3,5,7,9,11,13,15,17,19,21,23,25,27,29,30,33,35\}$ \\
			26 & $0.01111400$ & $\{1,3,5,7,9,11,13,15,17,19,21,23,25,27,28,31,33,35\}$ \\
			24 & $0.01058800$ & $\{1,3,5,7,9,11,13,15,17,19,21,23,25,26,29,31,33,35\}$ \\
			22 & $0.00972800$ & $\{1,3,5,7,9,11,13,15,17,19,21,23,24,27,29,31,33,35\}$ \\
			20 & $0.00966800$ & $\{1,3,5,7,9,11,13,15,17,19,21,22,25,27,29,31,33,35\}$ \\
			18 & $0.00954600$ & $\{1,3,5,7,9,11,13,15,17,19,20,23,25,27,29,31,33,35\}$ \\
			16 & $0.00945800$ & $\{1,3,5,7,9,11,13,15,17,18,21,23,25,27,29,31,33,35\}$ \\
			14 & $0.00940200$ & $\{1,3,5,7,9,11,13,15,16,19,21,23,25,27,29,31,33,35\}$ \\
			12 & $0.00936800$ & $\{1,3,5,7,9,11,13,14,17,19,21,23,25,27,29,31,33,35\}$ \\
			10 & $0.00936200$ & $\{1,3,5,7,9,11,12,15,17,19,21,23,25,27,29,31,33,35\}$ \\
			8  & $0.00935600$ & $\{1,3,5,7,9,10,13,15,17,19,21,23,25,27,29,31,33,35\}$ \\
			6  & $0.00935000$ & $\{1,3,5,7,8,11,13,15,17,19,21,23,25,27,29,31,33,35\}$ \\
			4  & $0.00934400$ & $\{1,3,5,6,9,11,13,15,17,19,21,23,25,27,29,31,33,35\}$ \\
			2  & $0.00934000$ & $\{1,3,4,7,9,11,13,15,17,19,21,23,25,27,29,31,33,35\}$ \\
			34 & $0.00933600$ & $\{2,5,7,9,11,13,15,17,19,21,23,25,27,29,31,33,35,36\}$ \\
			33 & $0.00420800$ & $\{2,4,6,8,10,12,14,16,18,20,22,24,26,28,30,32,34,35\}$ \\
			31 & $0.00295400$ & $\{2,4,6,8,10,12,14,16,18,20,22,24,26,28,30,32,33\}$ \\
			29 & $0.00209400$ & $\{2,4,6,8,10,12,14,16,18,20,22,24,26,28,30,31,34\}$ \\
			27 & $0.00178000$ & $\{2,4,6,8,10,12,14,16,18,20,22,24,26,28,29,32,34\}$ \\
			25 & $0.00125400$ & $\{2,4,6,8,10,12,14,16,18,20,22,24,26,27,30,32,34\}$ \\
			23 & $0.00039400$ & $\{2,4,6,8,10,12,14,16,18,20,22,24,25,28,30,32\}$ \\
			21 & $0.00033400$ & $\{2,4,6,8,10,12,14,16,18,20,22,23,28,30\}$ \\
			19 & $0.00021200$ & $\{2,4,6,8,10,12,14,16,18,20,21,24,30\}$ \\
			17 & $0.00012400$ & $\{2,4,6,8,10,12,14,16,18,19,22,24\}$ \\
			15 & $0.00006800$ & $\{2,4,6,8,10,12,14,16,17,20,22,24\}$ \\
			13 & $0.00003400$ & $\{2,4,6,8,10,12,14,15,18,20,24\}$ \\
			11 & $0.00002800$ & $\{2,4,6,8,10,12,13,24\}$ \\
			9  & $0.00002200$ & $\{2,4,6,8,10,11,14,16,18\}$ \\
			7  & $0.00001600$ & $\{2,4,6,8,9,12,14,16\}$ \\
			5  & $0.00001000$ & $\{2,4,6,7,10,12,14,16,18\}$ \\
			3  & $0.00000600$ & $\{2,4,5,8,10,12,14,16\}$ \\
			0  & $0.00000200$ & $\{1,2\}$ \\
			1  & $0.00000200$ & $\{2,3,6\}$ \\
			\bottomrule
		\end{tabular}
		\caption{Agents ordered by decreasing weight in the feasible certificate for the 37-point instance.}
		\label{tab:metric37-agents-by-weight}
	\end{table}

\end{document}